\documentclass[acmsmall, screen, nonacm]{acmart}
\AtBeginDocument{%
  }

\usepackage{graphicx}
\usepackage{subcaption}
\usepackage{wrapfig}
\usepackage[export]{adjustbox}
\usepackage{multirow}
\usepackage{booktabs}
\usepackage{threeparttable}

\usepackage{interval}
\intervalconfig{soft open fences}

\usepackage{cleveref}

\usepackage{tcolorbox}

\usepackage{mathtools}
\DeclarePairedDelimiter\rbra{\lparen}{\rparen}

\DeclarePairedDelimiter\cbra{\{}{\}}
\DeclarePairedDelimiter\abs{\lvert}{\rvert}
\DeclarePairedDelimiter\Abs{\lVert}{\rVert}
\DeclarePairedDelimiter\ceil{\lceil}{\rceil}
\DeclarePairedDelimiter\floor{\lfloor}{\rfloor}

\DeclarePairedDelimiter\sem{\llbracket}{\rrbracket}

\DeclareMathOperator{\tr}{tr}

\usepackage[thinc]{esdiff}

\newcommand{\ii}{\mathrm{i}}

\newcommand{\PREPARE}{\texttt{PREPARE}}
\newcommand{\SELECT}{\texttt{SELECT}}
\newcommand{\PREPAREC}{\texttt{PREPAREC}}
\newcommand{\SELECTC}{\texttt{SELECTC}}

\newcommand{\LindFront}{\textsf{LindFront}\@}
\newcommand{\ChannelIR}{\textsf{ChannelIR}\@}

\usepackage{tikz}
\usetikzlibrary{arrows.meta, positioning, shapes.geometric, calc}
\usetikzlibrary{quantikz2}
\usepackage{svg}
\usepackage{xcolor}

\usepackage{amsthm}
\newtheoremstyle{problem-style}
  {1ex}
  {1ex}
  {\normalfont}
  {}
  {\itshape}
  {.}
  {0.5em}
  {\thmname{#1}  \thmnumber{#2}  \thmnote{(#3)}}
\theoremstyle{problem-style}
\newtheorem{problem}{Problem}[section]

\usepackage{algorithm}
\usepackage{algpseudocode}
\floatname{algorithm}{Algorithm}

\begin{document}

\title{A Compilation Framework for Quantum Simulation of Non-unitary Dynamics}


\author{Qifan Huang}
\authornote{These authors contributed equally to this work.}
\email{huangqf@ios.ac.cn}
\affiliation{
    \department{Key Laboratory of System Software (Chinese Academy of Sciences) and State Key Laboratory of Computer Science}
    \institution{Institute of Software, Chinese Academy of Sciences}
    \city{Beijing}
    \country{China}
}
\affiliation{%
    \institution{University of Chinese Academy of Sciences}
    \city{Beijing}
    \country{China}
}
\author{Minbo Gao}
\authornotemark[1]
\email{gaomb@ios.ac.cn}
\affiliation{
    \department{Key Laboratory of System Software (Chinese Academy of Sciences) and State Key Laboratory of Computer Science}
    \institution{Institute of Software, Chinese Academy of Sciences}
    \city{Beijing}
    \country{China}
}
\affiliation{%
    \institution{University of Chinese Academy of Sciences}
    \city{Beijing}
    \country{China}
}
\author{Li Zhou}
\authornote{Corresponding author: Li Zhou, Mingsheng Ying}
\email{zhou31416@gmail.com}
\affiliation{
    \department{Key Laboratory of System Software (Chinese Academy of Sciences) and State Key Laboratory of Computer Science}
    \institution{Institute of Software, Chinese Academy of Sciences}
    \city{Beijing}
    \country{China}
}

\author{Mingsheng Ying}
\authornotemark[2]
\email{mingsheng.ying@uts.edu.au}
\affiliation{%
    \institution{University of Technology Sydney}
    \city{Sydney}
    \country{Australia}
}
\begin{abstract}
Most quantum compilers assume programs are reversible unitary circuits. This fits closed-system algorithms, but not open-system simulation, where the natural program objects are quantum channels describing non-unitary dynamics. We present a channel-first compilation framework that treats channels as first-class compilation objects. Our core IR, \ChannelIR, represents channels explicitly in Kraus form, a standard channel representation, with Pauli-sum structure, enabling algebraic rewrites before circuit synthesis. We instantiate the framework with \LindFront, a frontend that lowers continuous-time Lindbladian generators to short-time channels, and a backend that compiles these channels to executable circuits with structure-aware optimizations. On Lindbladian and channel-simulation benchmarks, the optimized pipeline reduces gate count by up to 99\% over an unoptimized channel-first baseline and scales better than circuit-first Stinespring compilation.
\end{abstract}




\maketitle

\section{Introduction}


Simulating quantum dynamics is widely regarded as one of the foundational applications of quantum computing~\cite{Fey82}.
Many quantum algorithms have been proposed for Hamiltonian simulation \cite{lloyd1996simulate,Lauvergnat2007simple,Childs2012hamiltonian,Childs2017efficient,childs2018andrew}, which models the time evolution of closed quantum systems governed by reversible (unitary) transformations.
To enable their practical applications, several programming and compilation frameworks for Hamiltonian simulation have been developed in the last few years \cite{Campbell_2019,Li2022paulihedral,Peng2024simuq,Cao2025marqsim}. More precisely, prior work has developed intermediate representations, synthesis procedures, and resource-aware lowering techniques that translate high-level Hamiltonian descriptions into executable quantum circuits. 

However, many physically relevant and algorithmically useful quantum processes do not correspond to closed systems and are not unitary. 
Realistic quantum systems interact with their environments, giving rise to dissipation, decoherence, and other irreversible effects.
These are commonly modeled as \emph{Markovian open-system dynamics}, which are characterized by Lindbladian master equations~\cite{GKS76,Lin76}.
Recent quantum algorithms increasingly target this setting, not only to simulate open physical systems more efficiently and accurately~\cite{KBG+11,Childs2017efficient,cleve2017lindblad,SHS+21,KSMM22,SHS+22,PW23a,PW23,DLL24,DVF+24,GKP+25,BM25,Chen2025thermal,PSZZ25,PSW25}, but also to exploit non-unitary dynamics as computational primitives for tasks such as state preparation~\cite{KB16,DCL24,Chen2025thermal,DLL25,LZL25,ZDH+25}, optimization~\cite{HLL+24,CLW+25}, and differential-equation solving~\cite{SGAZ25}.

These developments suggest an immediate motivating question: 

{\vskip 3pt}



\textbf{Question}: \textit{How should open-system quantum algorithms be programmed and compiled into executable quantum circuits?}

{\vskip 3pt}

Lindbladian simulation is a particularly important instance of this question, because it exposes the core difficulty sharply: the natural objects at the algorithmic level are no longer just unitary operators, but quantum channels describing the non-unitary open-system dynamics and generators of channels. Apart from Lindbladian simulation, more broadly, a growing number of channel-based quantum algorithms are appearing~\cite{LW23,RDB24,Dut24,TW25,LTL+25}, which also encounter this problem: program abstraction used by current quantum compilers no longer matches the semantic structure of the algorithms being compiled.

Current quantum compilation infrastructures are still largely organized around the unitary circuit model. 
In this model, programs are represented as sequences of reversible unitary gates acting on qubits, and compilation proceeds by lowering such circuits into sequences of gates in a target gate set.
This has been highly successful for closed-system algorithms, but it provides little direct support for algorithms whose natural descriptions are quantum channels rather than unitary circuits. 
In current workflows, channel-based algorithms must therefore be manually reduced to circuit-level constructions before compilation can proceed. 
This creates an \emph{abstraction mismatch}: the source program is naturally specified in terms of channels or generators, while the compiler expects a unitary circuit as input.

This mismatch has several consequences. 
First, it introduces an \emph{expressiveness gap}: programmers must encode channel semantics indirectly through circuit constructions, rather than describing the intended transformation directly. 
Second, it prevents the compiler from performing algebraic reasoning at the channel level. 
Once a channel has been  implemented as a circuit, structural information such as Kraus operators is no longer explicit, making many simplifications difficult or impossible to recover. 
Third, naive channel-to-circuit realizations often incur substantial resource overhead, for example by introducing large dilation-based constructions, excessive ancilla use, or heavily controlled multiplexed circuit structure. 
Thus, the main limitation of existing approaches is not merely that non-unitary algorithms are harder to implement, but that current compilers lack the right intermediate abstraction for representing and optimizing them. 

\paragraph{\textbf{This work.}} In this work, we propose that supporting channel-based quantum algorithms requires a different compilation perspective: 
\begin{center}
    \emph{Channels should be treated as first-class compilation objects.} 
\end{center}
Rather than forcing programmers to lower channels into circuits at the outset, a compiler should preserve channel structure explicitly, optimize it at the IR level, and postpone circuit realization until later stages. 
This is the central idea of our work: \emph{a channel-first compilation framework} for quantum programs whose natural semantics are given by quantum channels.

To realize this idea, we introduce \ChannelIR, a hierarchical intermediate representation for quantum channels\footnote{Our \ChannelIR{} can represent completely positive superoperators, which is a slight generalization of quantum channels.}. 
\ChannelIR{} represents channels explicitly in Kraus form, with Kraus operators built as block-encodings of Pauli-sum expressions or optional user-specified block-encoding primitives. 
This design is inspired in part by Cobble’s treatment of block-encodings~\cite{yuan2025cobblecompilingblockencodings}, but is tailored to channel-first compilation.
This representation is expressive enough to capture a broad class of channel-based workloads while preserving the algebraic structure needed for compiler reasoning. 
Because channels remain explicit in the IR, the compiler can apply rewrite rules and structural simplifications before any commitment is made to a concrete circuit implementation. 
At the same time, the representation aligns naturally with block-encoding- and linear-combination-of-unitaries (LCU) based circuit implementation methods, which are central to implementing non-unitary operators on quantum hardware.

To demonstrate the framework on an important and challenging workload, we instantiate it for Lindbladian simulation through a frontend called \LindFront. \LindFront{} accepts a Lindbladian generator specified by a Hamiltonian together with jump operators, and lowers it into a short-time channel approximation represented in \ChannelIR. This instantiation is important for two reasons. First, it shows how the framework handles a practically relevant class of open-system algorithms. Second, it makes concrete the broader compiler problem addressed in this paper: continuous-time generator descriptions must first be transformed into discrete channels, and those channels must then be compiled efficiently into circuits. 
By placing this process within a channel-first pipeline, our framework provides an end-to-end path from open-system models to optimized circuit implementations.

\paragraph{\textbf{Contributions and organization of the paper}} Our contributions are summarized as follows:
\begin{itemize}
    \item  Sec. 4:  \ChannelIR, a hierarchical intermediate representation for quantum channels in which Pauli sums and Kraus operators are building blocks, enabling algebraic reasoning via rewrite rules before circuit implementation; and a compilation path from \ChannelIR{} to executable quantum circuits via LCU and channel-LCU techniques.
    \item  Sec. 5: Structure-aware optimizations for the LCU circuit patterns arising in channel compilation, reducing control overhead, ancilla use, and non-Clifford cost.
    \item  Sec. 6: \LindFront, a frontend for Lindbladian simulation that transforms continuous-time Lindbladian generators into \ChannelIR, including first-order short-time approximation and higher-order series-expansion methods.
    \item Sec. 7: An evaluation on benchmarks for Lindbladian simulation and algorithm-derived quantum channels, showing that the channel-first baseline already outscales circuit-first Stinespring compilation, while the fully optimized pipeline further reduces gate count by $94.9\%$--$99.1\%$ and end-to-end compilation latency by $97.6\%$--$99.4\%$, and validates the correctness and scalability of the framework.

\end{itemize}




\section{Background}

We review some basic notions in quantum mechanics and quantum computing. For more detailed expositions, we refer readers to~\cite{NC10}.

\subsection{Models of Quantum Systems}

\paragraph{Closed quantum systems.} 

For a closed quantum system, its state space corresponds to a (complex) Hilbert space, i.e., a complete inner product space, which we usually denote as $\mathcal{H}$. In the study of quantum computation, it is usually assumed that the state space is finite-dimensional.
A pure state of a quantum system corresponds to a unit column vector in a Hilbert space.
We use the Dirac notation of kets like $\ket{\psi}$ to denote a pure state. Specifically, we use $\ket{j}$ to denote the unit column vector with only the $j$-th coordinate being $1$ and the other coordinates being $0$.

A common example is the \emph{qubit}, which is a $2$-dimensional quantum system. Generally speaking, the state of a qubit can be represented as $\ket{\psi} = \alpha\ket{0}+\beta\ket{1}$, for some $\alpha, \beta \in \mathbb{C}$ satisfying $\abs{\alpha}^2+\abs{\beta}^2 = 1$, and $\ket{0} = (1,0)^{\intercal}$ and $\ket{1} = (0,1)^{\intercal}$, where $^\intercal$ stands for the transpose operation.

The evolution of a closed quantum system is described by a unitary matrix $U$, satisfying $UU^{\dagger} = U^{\dagger}U$ = I, where $^{\dagger}$ stands for the conjugate transpose. For a quantum system starting at the state $\ket{\psi}$, the state after the evolution described by $U$ is just $U\ket{\psi}$. 
The following \emph{Pauli matrices} are famous examples of unitary matrices which are frequently used in quantum computing. They are also Hermitian (a matrix $A$ is Hermitian if $A = A^{\dagger}$).
\[
I = 
\begin{pmatrix}
    1 & 0 \\
    0 & 1
\end{pmatrix}, \quad
X = 
\begin{pmatrix}
    0 & 1 \\
    1 & 0
\end{pmatrix}, \quad
Y = 
\begin{pmatrix}
    0 & -i \\
    i & 0
\end{pmatrix}, \quad
Z = 
\begin{pmatrix}
    1 & 0 \\
    0 & -1
\end{pmatrix}.
\]

Measurement of a quantum system is a manipulation of a quantum system that yields a  result and changes the state of the system. In this paper, we will only consider \emph{projective measurements}.  A projective measurement $\mathcal{M}$ is a set of projectors $\cbra{P_i}_{i=1}^n$, where $P_i$ are orthogonal projections (i.e., $P_i = P_i^{\dagger}$ and $P_i^2 = P_i$) satisfying $\sum_{i=1}^n P_i = I$. For a quantum system in the state $\ket{\psi}$, applying measurement $\mathcal{M} = \cbra{P_i}_{i=1}^n$ will incur measurement result $i$ with probability $p_i = \Abs{P_i \ket{\psi}}^2$, leaving the post-measurement state $P_i\ket{\psi}/\sqrt{p_i}$ for $p_i > 0$.

For two quantum systems with state space $\mathcal{H}_1$ and $\mathcal{H}_2$ respectively, the composite system is described by the tensor product $\mathcal{H}_1\otimes \mathcal{H}_2$. If the first quantum system is in state $\ket{\psi}$ and the second in $\ket{\phi}$, the state of the composite system is described by $\ket{\psi}\otimes \ket{\phi}$, or $\ket{\psi}\ket{\phi}$ in short.
We use $(a_{ij})_{i,j=1}^N$ to denote the $N\times N$ matrix with $ij$-th entry being $a_{ij}$.
For an $m\times n$ matrix $A = (a_{rs})$ and $p\times q$ matrix $B = (b_{vw})$, recall that their tensor product is an $mp\times nq$ matrix with $(A \otimes B)_{p(r-1)+v,\;q(s-1)+w} = a_{rs} b_{vw}$.

\paragraph{Open quantum systems.} When one wants to describe quantum systems that may entangle with other physical systems (open quantum systems), one will use density matrices to describe mixed ensembles of states (mixed states in short). A density matrix $\rho$ is a positive semidefinite matrix with trace $\tr(\rho) = 1$, where the trace $\tr(\cdot)$ of a matrix equals to the sum of its diagonal entries. The set of all density operators on $\mathcal{H}$ is denoted as $\mathcal{D}\rbra{\mathcal{H}}$.

General operations in open quantum systems—including unitary evolution and measurements—are modeled as \emph{quantum channels}, which are completely positive and trace-preserving \emph{superoperators}. A superoperator is a linear map that maps operators into operators. A superoperator $\mathcal{E}: \mathcal{D}\rbra{\mathcal{H}} \to \mathcal{D}\rbra{\mathcal{H}}$ is called completely positive if for any Hilbert space $\mathcal{H}_{aux}$ and $\rho \in \mathcal{D}\rbra{\mathcal{H}\otimes \mathcal{H}_{aux}}$, $\rbra{\mathcal{E}\otimes I}\rbra{\rho}$ is positive;
it is called trace-preserving if for every $\rho \in \mathcal{D}\rbra{\mathcal{H}} $, $\tr\rbra{\mathcal{E}(\rho)}= \tr\rbra{\rho}$.
Completely positive superoperators, which we will simply call superoperators throughout the remainder of the paper, can be represented in \emph{Kraus form}:
a superoperator $\mathcal{E}$ can always be written as
$\mathcal{E}\rbra{\rho} = \sum_j K_j \rho K_j^{\dagger} $
for some matrices $K_j$s. The trace-preserving condition for it is equivalent to $\sum_j K_j^{\dagger}K_j = I$.


For the case where the dynamics of the quantum system is \emph{Markovian}, i.e., the future state depends only on the current state instead of the entire history, the state of the system can be characterized by the following \emph{Lindbladian master equation}:
\[
\diff{\rho}{t} = \mathcal{L}\rbra{\rho} = -\ii [\rho, H] + \sum_j \rbra*{L_j\rho L_j^{\dagger} - \frac{1}{2}\cbra*{L_j^{\dagger}L_j, \rho}},
\]
where $[A,B]\triangleq AB-BA$ denotes the commutator operation, $\cbra{A,B}\triangleq AB+BA$ denotes the anti-commutator operation,  $H$ is a Hermitian matrix representing the \emph{system Hamiltonian} describing the unitary aspect of the dynamics, and the matrices $L_j$s serve as \emph{jump operators} describing the dissipative part of the dynamics.


\subsection{Quantum Linear Algebra}

\begin{definition}[Block-Encoding]\label{def:block-encoding}
    Let $A\in \mathbb{C}^{2^n\times 2^n}$ be a matrix. For $\alpha, \varepsilon >0$ and a positive integer $a$, an $(n+a)$-qubit unitary $U$ is called an $(\alpha, a,\varepsilon)$-block-encoding of $A$, if
    \[
    \Abs{\alpha\bra{0}^{\otimes a}U \ket{0}^{\otimes a} -A }\le \varepsilon,
    \]
    where $\Abs{\cdot}$ denotes the operator norm of matrix.
\end{definition}
If $\varepsilon = 0$ and $a$ is clear from the context, we will simply call it an $\alpha$-block-encoding. For $\alpha=1$, we will directly call it a block-encoding and use the notation $\mathcal{B}[A]$ to denote it.
In simple words, this means
\[
U \approx 
\begin{pmatrix}
    A & *\\
    * & *
\end{pmatrix}.
\]
Block-encoding provides a convenient way to represent a general matrix as a sub-block of a unitary for implementing them in quantum circuits.


To implement block-encodings of operators expressed as sums of simpler terms, we will rely on the linear-combination-of-unitaries (LCU) technique shown below.

\begin{theorem}[Linear Combination of Block-Encodings]\label{thm:LCU}
    For an $m$-dimensional vector $y = (y_0, \ldots, y_{m-1}) \in \mathbb{C}^m$ with $\sum_j \abs{y_j}\le \beta$, suppose there exists a pair of $b$-qubit unitaries ($b\ge \ceil{\log m}$ is an integer) $\PREPARE_L$ and $\PREPARE_R$ such that
    \[
    \begin{aligned}
        \PREPARE_L\ket{0}^{\otimes b} &= \sum_{j=0}^{2^b-1}c_j\ket{j} ,\quad
        \PREPARE_R\ket{0}^{\otimes b} = \sum_{j=0}^{2^b-1} d_j\ket{j}, \\
        \sum_{j=0}^{m-1}\abs*{\beta c_j^* d_j - y_j} &\le \varepsilon_1,
        \quad \text{and} \quad c_j^*d_j  = 0, \;\forall j=m, \dots, 2^b-1.
    \end{aligned}
    \]
    Let $A_0, A_1, \ldots, A_{m-1}$ be $s$-qubit operators, and $U_j$ is an $\rbra{\alpha, a,\varepsilon_2}$-block-encoding of $A_j$ for each $j$.  We use $\SELECT$ to denote the $(s+ a+ b)$-qubit unitary
    \[
    \SELECT = \sum_{j=0}^{m-1} \ket{j}\bra{j}\otimes U_j + \rbra*{\rbra*{I-\sum_{j=0}^{m-1}\ket{j}\bra{j}}\otimes I^{\otimes (s+a)}}.
    \]
Then  $\rbra{\PREPARE_L^{\dagger}\otimes I^{\otimes (s+a)}} \SELECT  \rbra{\PREPARE_R\otimes I^{\otimes (s+a)}}$
    is an $(\alpha\beta, a+b, \alpha\varepsilon_1+\alpha\beta\varepsilon_2)$-block-encoding
    of $\sum_{j=0}^{m-1} y_j A_j$.
\end{theorem}





\section{Motivating Example: a Simple Lindbladian Simulation Compilation}


To illustrate the compilation challenges and our contributions, we consider simulating  simple Markovian dynamics of an open quantum system.
Unlike closed quantum systems, whose evolution is described by \emph{reversible} unitary operators, Markovian open-system dynamics are governed by Lindbladian generators that induce quantum channels (completely positive trace-preserving maps) describing \emph{irreversible} non-unitary processes such as dissipation and decoherence. 
Our goal is not to analytically solve the resulting differential equation, but rather to compile the approximation of this dynamical evolution into executable quantum circuits. 

\subsection{Problem Setup}


Consider the following textbook example that describes the thermal relaxation of a single qubit interacting with an environment~\cite{Bre07}. This is a single-qubit open system whose dynamics are described by two jump operators $L_1 = \sqrt{\gamma(N+1)}\ket{1}\bra{0}$ and $ L_2 = \sqrt{\gamma N}\ket{0}\bra{1}$ for some $\gamma, N >0$.
Our goal is to compile a simulation of the dynamical evolution $\mathcal{E}_t = e^{t\mathcal{L}}$ into executable quantum circuits.

Although this example only involves a single qubit, the same compilation pattern appears in larger Lindbladian systems with more terms, and each jump operator could expand into large numbers of Pauli terms.



\subsection{Compilation Process}

\subsubsection{Input form}
Our compilation pipeline begins by translating the input Lindbladian generator into our intermediate representation \ChannelIR.
We require the input Lindbladian to be provided as a Hamiltonian $H$ together with a set of jump operators $L_j$, both specified as  \emph{Pauli sums} or explicit matrices. 

For the input operators in matrix form, the compiler will decompose them into Pauli sums as
\[
    H = \sum_k \alpha_k V_{0k},  \quad L_j = \sum_{k=0}^{q-1} \beta_{jk} V_{jk},
\]
where $V_{jk}$ is an $n$-fold tensor product of Pauli matrices $\{I,X,Y,Z\}$ with a phase $e^{\ii \theta}$, $\beta_{jk}$ is a non-zero complex number. 
For the example above, the jump operators can be represented as
\[
 L_1 = \frac{\sqrt{\gamma(N+1)}}{2} X +\frac{\sqrt{\gamma(N+1)}}{2} \rbra*{\ii Y},  \quad L_2 = \frac{\sqrt{\gamma N} }{2} X + \frac{\sqrt{\gamma N} }{2} \rbra*{-\ii Y}.
\]


\subsubsection{\LindFront{}: from Lindbladian generators to superoperators represented in \ChannelIR}
The Lindbladian generator specifies a continuous-time evolution, while quantum circuits implement \emph{discrete} operations. The compiler therefore applies a lowering pass that approximates the evolution over a small timestep $\delta$.

Our framework supports multiple lowering strategies. One approach is the first-order expansion method~\cite{cleve2017lindblad}. 
\LindFront{} computes the Kraus form of a superoperator that approximates the short-time evolution by its first-order expansion as $e^{\delta\mathcal{L}}\approx \mathcal{I} +\delta\mathcal{L}$. 
Our framework also supports higher-order expansion studied in~\cite{Li2023simulating}.

The Kraus operators of the superoperator are therefore computed as follows:
\begin{equation}
A_0 = I - \frac{\delta}{2}\sum_{j=1}^{2}L_j^{\dagger}L_j, \quad A_j = \sqrt{\delta} L_j, \ j= 1,2. 
\end{equation}

\subsubsection{\ChannelIR{} simplifications}
\ChannelIR{} provides a set of rewrite rules for simplifying Kraus operators and Pauli-sum expressions for matrices. For the example above, the result after rewriting is a superoperator represented in Kraus form within \ChannelIR:
\begin{figure}[ht]
\vspace{-1em}
\[
\begin{array}{l}
\texttt{channel}\ \mathcal{E}_\delta \ \{ 
\quad \ A_0 \;=\; \left(1-\frac{\delta\gamma}{4}(2N+1)\right) I
-\frac{\delta\gamma}{4} Z,   
\quad \ A_1 \;=\; 
    \sqrt{\tfrac{\delta\gamma(N+1)}{2}} \cdot X
    \;+\;
    \sqrt{\tfrac{\delta\gamma(N+1)}{2}} \cdot (\ii Y)  \\[2pt]
\quad \ A_2 \;=\;
    \sqrt{\tfrac{\delta\gamma N}{2}} \cdot X
    \;+\;
    \sqrt{\tfrac{\delta\gamma N}{2}} \cdot (-\ii Y) 
\}
\end{array}
\]
\vspace{-1em}
\caption{The resulting channel represented in Kraus form within \ChannelIR.}
\label{fig:lindir-pseudocode}
\end{figure}



\subsubsection{Circuit implementation: from \ChannelIR{} to quantum circuits}

The next stage lowers the Kraus form representation of the superoperator stored in \ChannelIR{} into executable quantum circuits. Given Kraus operators $A_j$, the compiler first synthesizes block-encodings (see Definition~\ref{def:block-encoding}) $\mathcal{B}[A_j]$ via the standard linear-combination-of-unitaries technique for each operator and then constructs a channel implementation using a \emph{channel} linear combination of unitaries (channel-LCU)~\cite{cleve2017lindblad,Li2023simulating}.

In our motivating example, the Kraus operators obtained in Step~2 are $A_0, A_1, A_2$, where $A_1$ and $A_2$ come from the two jump operators, and $A_0$ is the no-jump term. After compiling each $A_j$ into a block-encoding, the channel-LCU construction combines these three implementations into a single circuit for the short-time evolution channel approximating $e^{\delta \mathcal{L}}$.

\begin{figure}[h]
\centering
\begin{subfigure}{0.35\linewidth}
    \centering
    \[ \resizebox{0.6\linewidth}{!}{
    \begin{quantikz}[row sep=0.25cm, column sep=0.4cm]
            & \gate[wires=1]{H} 
            & \octrl{1}
            & \ctrl{1}
            & \gate[wires=1]{H} 
            & \qw \\
            & \qw 
            & \gate{X}
            & \gate{\ii Y}
            & \qw 
            & \qw \\
    \end{quantikz}}
    \]
\\ 
 \[ \resizebox{0.6\linewidth}{!}{
    \begin{quantikz}[row sep=0.25cm, column sep=0.4cm]
            & \gate[wires=1]{H} 
            & 
            & \ctrl{1}
            & \gate[wires=1]{H} 
            & \qw \\
            & \qw 
            & \gate{X}
            &  \gate{Z} 
            & \qw 
            & \qw \\
    \end{quantikz}}
    \]
    \caption{The block-encoding circuit $\mathcal{B}[A_1]$ via LCU (top) and its optimized  form (bottom). $H = \frac{1}{\sqrt{2}}(X+Z)$ is the Hadamard gate.}
    \label{fig:lcu-block-encoding-motivating}
\end{subfigure}
\hfill
\begin{subfigure}{0.6\linewidth}
\resizebox{0.85\textwidth}{!}{
    \begin{quantikz}
\lstick{$\ket{0}$}  &\gate[2]{\PREPARE}&\octrl{2} & \octrl{2} & \ctrl{2} &  \\
\lstick{$\ket{0}$}  && \octrl{1}& \ctrl{1} &\octrl{1}&   \\
\lstick{$\ket{0}$}    & & \gate[2, style={fill=orange!20}]{\mathcal{B}[A_0]} & \gate[2, style={fill=orange!20}]{\mathcal{B}[A_1]} & \gate[2, style={fill=orange!20}]{\mathcal{B}[A_2]} &  \\
&  &&&& \\
\end{quantikz}}
\caption{The circuit approximating $e^{\delta\mathcal{L}}$ via channel-LCU.}
\end{subfigure}
\end{figure}


However, the naive channel-LCU implementation introduces a large number of multi-controlled operations even for this simple example. 
These structures dominate the cost of the circuit, motivating our next step which introduces structure-aware optimizations to significantly reduce this overhead.

\paragraph{Optimization: optimizing \SELECT{} in block-encodings/LCUs.}
When constructing a \SELECT{} oracle for a linear combination of 
$N$ Pauli strings, the straightforward approach involves implementing 
$N$ multi-controlled gates, each with $\log N$ control qubits. 
For instance, when considering the block-encoding of the Kraus operator $A_1$, 
the naive \SELECT{} oracle is a two-branch multiplexor that applies $X$ and $iY$ on each branch. Since we have $iY = ZX$, they share a common factor $X$. 
We may factor out this common $X$ operator and rewrite the circuit so that only the remaining $Z$ operates depending on the control. 
As shown in Figure~\ref{fig:lcu-block-encoding-motivating}, these two circuits are equivalent, but the rewritten circuit reduces the control overhead. This simple example foreshadows the more general Pauli-structure-aware optimization developed later.

We summarize the above compilation workflow in~\Cref{fig:workflow}. The details of each step are introduced in subsequent sections.
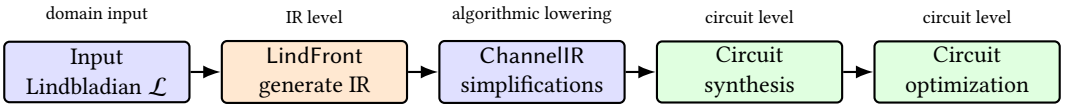
\begin{figure}[ht]
\centering
\begin{tikzpicture}[
    >=Latex,
    node distance=4mm and 4mm,
    font=\small,
    box/.style={
        draw,
        rounded corners=2pt,
        thick,
        minimum height=8mm,
        minimum width=2.45cm,
        align=center,
        inner sep=2pt
    },
    lind/.style={box, fill=blue!12},
    ir/.style={box, fill=orange!18},
    circ/.style={box, fill=green!12},
    arr/.style={->, thick}
]

\node[lind] (input) {Input\\Lindbladian $\mathcal{L}$};
\node[ir, right=of input] (irnode) {\LindFront{}\\
generate IR};
\node[lind, right=of irnode] (lower) {\ChannelIR\\
simplifications};
\node[circ, right=of lower] (syn) {Circuit\\synthesis};
\node[circ, right=of syn] (opt) {Circuit\\optimization};

\draw[arr] (input) -- (irnode);
\draw[arr] (irnode) -- (lower);
\draw[arr] (lower) -- (syn);
\draw[arr] (syn) -- (opt);

\node[above=1mm of input, font=\scriptsize] {domain input};
\node[above=1mm of irnode, font=\scriptsize] {IR level};
\node[above=1mm of lower, font=\scriptsize] {algorithmic lowering};
\node[above=1mm of syn, font=\scriptsize] {circuit level};
\node[above=1mm of opt, font=\scriptsize] {circuit level};

\end{tikzpicture}
\caption{
Compilation workflow for the motivating example.
}
\label{fig:workflow}
\end{figure}


\section{\ChannelIR}
In this section, we propose \ChannelIR, an intermediate representation specifically designed for quantum channels.


\subsection{Motivation and Design Goals}
Quantum channels provide a natural abstraction for describing open-system quantum dynamics. 
However, existing quantum compilation infrastructures are primarily organized around unitary circuit representations, which assume reversible closed-system evolution. 
When algorithms are expressed in terms of quantum channels or generators (such as Lindbladians), they must typically be manually translated into circuit-level constructions before compilation can proceed.
This translation prematurely fixes a circuit implementation for channels and obscures the algebraic structure of the original channels.

\ChannelIR{} is designed to bridge this gap by introducing an intermediate representation where quantum channels are treated as first-class compilation objects. The design of \ChannelIR{} is guided by the following goals.

\textit{Preserve Channel-Level Structure}
Many quantum algorithms involving open systems are naturally expressed as quantum channels, rather than as sequences of unitary gates. For example, Lindbladian simulation produces channels whose Kraus operators describe dissipative processes. Representing these transformations directly at the IR level preserves their mathematical structure and avoids prematurely expanding them into large circuit constructions.

\ChannelIR{} represents quantum channels explicitly in a Kraus form, enabling the compiler to manipulate channels symbolically before lowering them into circuit implementations.

\textit{Enable Algebraic Reasoning Prior to Circuit Synthesis.}
Kraus representations of a quantum channel admit multiple equivalent decompositions, and algebraic transformations can often simplify them substantially. 
However, once a channel has been lowered into a quantum circuit representation, these opportunities are difficult to recover.

\ChannelIR{} retains the operator-level representation of Kraus operators as matrix sums of Pauli strings, allowing compiler passes to apply algebraic rewrites such as
\begin{itemize}
    \item eliminating redundant Kraus operators,
    \item merging equivalent operator terms,
    \item removing irrelevant global phases, and
    \item applying unitary transformations for Kraus terms.
\end{itemize}
These transformations can reduce the number of Kraus operators in the representation, thereby reducing circuit cost.

\paragraph{Preserve Linear Structure for Block-Encoding-Based Synthesis.}

For implementing non-unitary operations, current approaches rely on the channel-LCU method based on block-encodings. Therefore, Kraus operators should be in a form that allows efficient constructions of block-encodings.

\ChannelIR{} explicitly represents Kraus operators as linear combinations of Pauli strings, which are directly implementable unitary operators. This representation aligns naturally with block-encoding and LCU synthesis methods, allowing the compiler to translate operator-level expressions into circuits systematically.

\subsection{Core Syntax}

We present the core syntax of \ChannelIR{} in \Cref{fig:channelir-grammar}.
















\begin{figure}[ht]
\vspace{-1em}
\centering
\setlength{\arraycolsep}{2pt} 
\[
\begin{array}{rcl@{\qquad}rcl}

\mathsf{Scalar} \; \beta  &\in& \mathbb{C} \setminus \cbra{0}
&
\mathsf{Phase} \; \theta &\in& \interval[open right]{0}{2\pi}
\\

\mathsf{PauliAtom} \; g_i & ::= & I \mid X \mid Y \mid Z
&
\mathsf{PauliString} \; G & ::= &  g_1g_2\cdots g_n
\\

\mathsf{PauliUnitary} \; v  & ::= & e^{\ii\theta} G
&
\mathsf{Primitive} \; P  & ::= &  v \mid  \texttt{blockenc}(q)
\\

\mathsf{Kraus} \; K  & ::= & \sum_j \beta_j P_j
&
\mathsf{Channel} \; \mathcal{C} & ::= & \texttt{channel} \{ K_1 ; \dots ; K_m \}
\end{array}
\]
\vspace{-1em}
\caption{Core syntax of \ChannelIR.}
\label{fig:channelir-grammar}
\end{figure}

For a superoperator in Kraus form $\mathcal{C}(\rho) = \sum_{j=1}^m K_j \rho K_j^{\dagger}$,
\ChannelIR{} represent it explicitly as the collection of Kraus operators $\texttt{channel}\{K_j\}$.

Each Kraus operator $K$ is expressed in a linear combination form as
$K = \sum_j \beta_j P_j$,
where $\beta_j$s are non-zero complex coefficients, and $P_j$s are primitive operators.
During the compilation, this linear combination form is kept via a rewrite rule (see rule \textbf{PS2} in~\Cref{sec:channelir-rewrite}): if a term $\beta_jP_j$ becomes $0\cdot P_j$ in the process, the term is removed from the sum.

Primitive operators represent the lowest-level operator constructs in the IR. \ChannelIR{} currently supports two kinds of primitives:
\begin{itemize}
    \item $\mathsf{PauliUnitaries}$, which are tensor products of $\mathsf{PauliAtoms}$  ($I,X,Y,Z$) with a global phase factor $e^{\ii \theta}$, and
    \item Block-encoding $\texttt{blockenc}(q)$, which denotes a user-specified implementation of a block-encoding of matrix $q$.
\end{itemize}

The term $\texttt{blockenc}(q)$ provides an alternative choice for users to specify block-encodings according to their own implementations. Treating block-encodings as primitives allows \ChannelIR{} to incorporate externally provided implementations directly, without requiring them to be expanded into Pauli decompositions.

\paragraph{Phase and Coefficient Flexibility}
\ChannelIR{} allows both complex coefficients in matrix sums and complex phase on primitive operators. In other words, phase factors can appear either in the coefficient $\beta_j$, or in the primitive operator $P_j$, or in both.

This redundancy is intentional. Different synthesis strategies may place phase factors in different locations of the generated circuit: within primitive block encodings, or in state-preparation unitaries used in LCU constructions. 
Allowing phases to appear either in coefficients or primitives therefore avoids prematurely committing to a particular implementation strategy, enabling later compilation to choose appropriate phase placement according to the target synthesis strategy.



\subsection{Dimension Consistency}

Beyond the above syntax, \ChannelIR{} imposes a small collection of structural
well-formedness conditions. These conditions ensure that every \ChannelIR{} term
denotes a dimensionally consistent superoperator and that subsequent rewrite
and lowering passes operate on a valid intermediate representation.

\paragraph{Dimensions of primitives.}
Each primitive operator $P$ in \ChannelIR{} is associated with the number of qubits that it operates on. We write $P : n$
to indicate that $P$ acts on an $n$-qubit Hilbert space.
For Pauli unitaries, this is determined by the length of the Pauli string.
In other words, 
if $P = e^{\ii \theta} G$ for some $\mathsf{PauliString}$ $G = g_1 g_2\dots g_n$, then $P:n$.
For a primitive of the form $\mathsf{blockenc}(q)$, the arity is determined by
the number of qubits that the encoded operator $q$ operates on.

\textit{Dimension of Kraus operators.}
A Kraus expression $K = \sum_j \beta_j P_j$
is well-typed only if all primitives $P_j$ appearing in the sum satisfy $P_j:n$. In that case, we write $K : n$.
This condition guarantees that linear combinations are formed only between
operators acting on the same Hilbert space.

\textit{Typing of channels.}
A channel term $C = \mathsf{channel}\{K_1,\dots,K_m\}$
is well-typed only if every Kraus operator $K_i$ has the same arity $n$.
We then write $C : n \to n$,
indicating that $C$ denotes a superoperator on $n$ qubits.


\textit{Zero and empty expressions.}
We allow the zero Kraus operator and the empty Kraus sum as intermediate forms
during rewriting. These expressions remain well-typed as long as their ambient
arity is clear from context. This convention is useful for local simplifications
such as zero-term elimination and zero-Kraus removal, while preserving a simple
algebraic presentation of the rewrite system.

\textit{Role in compilation.}
The above conditions serve two purposes. First, they ensure that the
denotational semantics of \ChannelIR{}  is defined only on dimensionally consistent terms. 
Second, they provide a basic invariant during compilation: rewrites and lowering should preserve the dimension of the system.

\subsection{Semantics}

\ChannelIR{}  represents superoperators in Kraus form, and we define its semantics as follows.
\paragraph{Pauli strings} 
For $G\in \mathsf{PauliString}$ with $G = g_1 g_2 \dots g_n$, its semantics 
is defined as  $\sem{G} = \otimes_{j=1}^n g_j$,
i.e., $\sem{G}$ is an $n$-qubit Pauli gate.

\textit{Pauli unitaries.}
For $u\in \mathsf{PauliUnitary}$ with $u = e^{\ii \theta}G$, we define its semantics as $\sem{u} = e^{\ii \theta}\sem{G}$.

\textit{Primitive operators.}
 For $P\in \mathsf{Primitive}$, $P$ is either a Pauli unitary or of the form 
 $\texttt{blockEnc}(q)$. 
 We define the semantics of $\texttt{blockenc}(q)$ as $\sem{\texttt{blockenc}(q)} = q$.
 
\textit{Kraus operators.} 
For $K\in \mathsf{Kraus}$, it is represented as a linear combination of primitives as
$K = \sum_j \beta_j P_j$
with $\beta_j \ne0$.
The semantics is defined as $\sem{K} = \sum_j \beta_j \sem{P_j}$.

\textit{Channels.}
For $\mathcal{C}\in \mathsf{Channel}$ with $\mathcal{C} = \texttt{channel}\{K_1,\dots,K_m\}$, its semantics is a completely positive linear map satisfying 
$\llbracket \mathcal{C} \rrbracket (\rho)=\sum_{i=1}^m \sem{K_i} \rho \sem{K_i}^\dagger$ for any $\rho$.


\begin{theorem}\label{thm:well-formed-IR-are-channels}
    If $\mathcal{C}:n\to n$, then $\sem{\mathcal{C}}$ is a (completely positive) superoperator from the set of $n$-qubit density matrices to the set of $n$-qubit density matrices.
\end{theorem}
\begin{proof}
    This is direct by induction on the structure of $\mathcal{C}$.
\end{proof}




\subsection{Algebraic Rewrite Rules}
\label{sec:channelir-rewrite}

\ChannelIR{}  supports a collection of algebraic rewrite rules that preserve the 
semantics of the represented quantum channel while simplifying its internal
structure. These rewrites enable the compiler to reduce the number of Kraus
operators and simplify the Pauli-sum expressions appearing in Kraus operators
prior to circuit synthesis.

We write $C \Rightarrow_R C'$ to denote that channel $C$ can be rewritten into
$C'$ using the set of rules presented in~\Cref{fig:rewrite-rules}.

\newcommand{\channel}[1]{\mathsf{channel}\{#1\}}
\newcommand{\epssum}{\epsilon}

\begin{figure}[h]
\centering
\footnotesize
\begin{tabular}{ll}
\toprule
\multicolumn{2}{l}{\textbf{Kraus-operator simplifications}} \\[3pt]
\midrule
\textbf{(PS1)} \textit{Term merging}
&
$\beta_1 P + \beta_2 P \Rightarrow (\beta_1+\beta_2)P$
\\[4pt]

\textbf{(PS2)} \textit{Zero-term elimination}
&
$0 \cdot P \Rightarrow \epssum$
\\
&
where $\epssum$ denotes the empty sum with
$\llbracket \epssum \rrbracket = 0$
\\[8pt]

\midrule
\multicolumn{2}{l}{\textbf{Channel-level Kraus simplifications}} \\[3pt]
\midrule

\textbf{(K1)} \textit{Zero Kraus elimination}
&
$\channel{K_1,\ldots,K_{j-1},0,K_{j+1},\ldots,K_m}
\Rightarrow
\channel{K_1,\ldots,K_{j-1},K_{j+1},\ldots,K_m}$
\\[6pt]

\textbf{(K2)} \textit{Global phase elimination}
&
$\channel{\ldots,e^{i\theta}K,\ldots}
\Rightarrow
\channel{\ldots,K,\ldots}$
\\[6pt]

\textbf{(C1)} \textit{Kraus permutation}
&
for any permutation $\pi$ \\
&
$\channel{K_1,\ldots,K_m}
\Rightarrow
\channel{K_{\pi(1)},\ldots,K_{\pi(m)}}$
\\[6pt]

\textbf{(C2)} \textit{Kraus unitary transform}
&
if $K'_j = \sum_k u_{jk}K_k$ where $U = (u_{jk}) \in \mathbb{C}^{m\times m}$ is unitary
\\
&
$\channel{K_1,\ldots,K_m}
\Rightarrow
\channel{K'_1,\ldots,K'_m}$
\\[6pt]

\textbf{(C2')} \textit{Two-Kraus unitary transform}
&
if $|a|^2 + |b|^2 = 1$
\\
&
$\channel{K_1,K_2,K_3,\ldots,K_m}
\Rightarrow
\channel{aK_1+bK_2,-b^*K_1+a^*K_2,K_3,\ldots,K_m}$
\\[6pt]

\textbf{(C3)} \textit{Kraus merging}
&
$\channel{a_1K,a_2K,\ldots,a_jK,K_1,\ldots,K_m}
\Rightarrow
\channel{\sqrt{\sum_{\ell=1}^j |a_\ell|^2}\,K,K_1,\ldots,K_m}$
\\[6pt]

\textbf{(C3')} \textit{Two-Kraus merging}
&
$\channel{a_1K,a_2K,\ldots}
\Rightarrow
\channel{\sqrt{|a_1|^2+|a_2|^2}\,K,\ldots}$\\
\bottomrule
\end{tabular}
\vspace{-0.5em}
\caption{Rewrite rules for \ChannelIR.}
\label{fig:rewrite-rules}
\end{figure}

The rule \textbf{C3} (or \textbf{C3'}) is actually a special case of the rule \textbf{C2} (or \textbf{C2'}).

\paragraph{Soundness.} We can show that rewrite rules do not change
the semantics.

\begin{theorem}\label{thm:soundness-of-rewrite-rules}
For the rewrite rules defined above and any \ChannelIR{} program $C$,
if $C \Rightarrow_R C'$, then $\sem{C} = \sem{C'}$.
\end{theorem}

\begin{proof}[Proof sketch]
This follows by induction on the rules and the unitary freedom in the operator-sum representation of quantum channels~\cite[Theorem 8.2]{NC10}.
\end{proof}


\subsection{Simplification Power of the Rewrite System}
The rewrite rules introduced above provide more than semantic preservation: they enable simplification of Kraus representations before circuit synthesis. Specifically, they allow us to:
\begin{itemize}
    \item reduce the number of Kraus operators representing a channel and
    \item simplify the internal Pauli-sum structure of individual Kraus operators.
\end{itemize}
We illustrate this with theory and examples below.


\begin{proposition}[Kraus-rank minimization by rewriting]\label{prop:kraus-rank-minimization}
Let $C = \mathrm{channel}\{K_1,\dots,K_n\}$
be a well-typed \ChannelIR{} program, and let $r$ be the minimal number of Kraus
operators among all Kraus representations of $\llbracket C \rrbracket$.
Then there exists a rewrite sequence $C \Rightarrow_R C'$
such that $C'$ has exactly $r$ nonzero Kraus operators.

Moreover, this sequence can be either chosen to use one application of rule \textbf{C3}, followed by one application of \textbf{C2}, and at most $n-r$ applications of rule \textbf{K1}; 
 or it can be chosen as one application of rule \textbf{C3}, followed by $O(n)$ applications of rule \textbf{C2'} and at most $n-r$ applications of rule \textbf{K1}.
In both cases, the number of applications is $O(n)$.
\end{proposition}


This proposition shows that the rewrite system is expressive enough to recover a minimal Kraus representation of a channel.
In practice, this means redundant Kraus terms produced by earlier compilation stages can be eliminated before circuit synthesis. 
Since the synthesis pipeline introduces block-encodings and channel-selection structures for each Kraus operator, reducing the Kraus count could directly reduce the number of ancillas for channel-LCU, and may reduce the number of elementary gates and ancillas in the final circuit.

Furthermore, the above rules can not only minimize the number of Kraus operators, but also simplify Kraus operators themselves to reduce block-encoding cost. We illustrate below.

\begin{example}[Unitary mixing for dephasing]
Consider the single-qubit dephasing channel $\mathcal{E}$, which is widely used in quantum computation~\cite{NC10}. It is usually written in this form:
\[
\mathcal{E}\rbra{\rho} = \ket{0}\bra{0}\rho\ket{0} \bra{0}+ \ket{1}\bra{1}\rho\ket{1} \bra{1}.
\]
Therefore, it is translated directly in our \ChannelIR{} as
$\mathcal{E} = \texttt{channel}\cbra*{I/2+Z/2, I/2-Z/2}$,
yielding two Kraus operators with $2$ Pauli terms in each operator.
We can apply Rule \textbf{C2'} with $a = b = \frac{1}{\sqrt{2}}$,
giving
$ \texttt{channel}\cbra*{I/2+Z/2, I/2-Z/2} \Rightarrow 
\texttt{channel}\cbra*{I/\sqrt{2},Z/\sqrt{2}}$,
which yields two Kraus operators with only $1$ Pauli term in each operator.

By the above rewriting, we can reduce both the number of gates and ancillas in the circuit.
\end{example}

\subsection{From \ChannelIR{} to Circuit Implementations}


\paragraph{Objective and interface}
The input to the circuit implementation stage is a superoperator $\mathcal{E}$ expressed in \ChannelIR{} in Kraus form as $\mathsf{channel}\cbra{A_1, A_2, \dots, A_n}$
with each $A_j = \sum_{k} \beta_{jk} P_{jk}$ in a Pauli-sum form or a block-encoding $\mathcal{B}[A_j]$ is given.
The output will be a description of a quantum circuit that implements this superoperator.

\paragraph{Implementation procedure}
The compiler constructs a corresponding quantum circuits in two steps, which follows the standard construction as in~\cite{cleve2017lindblad,Li2023simulating}.
\begin{itemize}
    \item Construct a block-encoding of each Kraus operator written as a Pauli sum.
    \item Combine these block-encodings using the channel-LCU construction.
\end{itemize}

This decomposition follows directly from the structure of the IR.

\textbf{Step 1: Synthesize a block-encoding for each Kraus operator.}
The first step is to construct a block-encoding $\mathcal{B}[A_j]$ for each Kraus operator $A_j$. In our representation, each $A_j$ is either given a block-encoding or is expressed as a linear combination of Pauli unitaries. For the latter case, as Pauli operators are unitary and they trivially serve as their own block-encodings, we can directly apply~\Cref{thm:LCU} to construct a block-encoding (see~\Cref{fig:lcu-block-encoding}). 

\begin{figure}[ht]
  \centering

  \begin{subfigure}[t]{0.51\linewidth}
    \centering
    \resizebox{\linewidth}{!}{%
    \begin{quantikz}
        & \gate{\PREPARE_R}
        & \ctrl{2}
        & \gate{\PREPARE_L^\dagger}
        & \qw \\
        & \qw
        & \gate[wires=2][2.0cm]{\SELECT}
        & \qw
        & \qw \\
        & \qw
        & \qw
        & \qw
        & \qw
    \end{quantikz}
    }
    \caption{LCU block-encoding circuit for the Pauli sum
    $\sum_{j=0}^{m-1} \beta_j P_j$.}
    \label{fig:lcu-block-encoding}
  \end{subfigure}\hfill
  \begin{subfigure}[t]{0.47\linewidth}
    \centering
    \resizebox{\linewidth}{!}{%
    \begin{quantikz}
      \lstick{$\ket{0}^{\otimes \ell}$}
        & \gate{\PREPAREC}
        & \ctrl{1}
        & \qw \\
      \lstick{$\ket{0}^{\otimes a}$}
        & \qw
        & \gate[2]{\SELECTC}
        & \meter{} \rstick{$\ket{0}^{\otimes a}$} \\
        & \qw
        & \qw
        & \qw
    \end{quantikz}
    }
    \vspace{-0.5em}
    \caption{Circuit structure for the channel-LCU construction.}
    \label{fig:channel-lcu}
  \end{subfigure}

  \caption{Comparison of LCU-style circuit constructions.}
  \label{fig:lcu-circuits}
\end{figure}
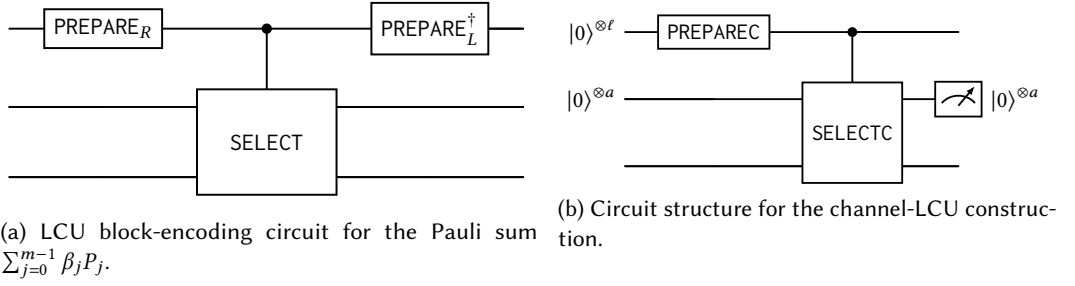


\textbf{Step 2: Lift from Kraus block-encodings to a channel via channel-LCU.}

After block-encodings for all Kraus operators are constructed, the compiler synthesizes the full channel using the channel-LCU method~\cite{cleve2017lindblad,Li2023simulating} which we outline below.

Given block-encodings $\mathcal{B}[A_j]$ as $\alpha_j$-block-encodings of $A_j$, setting
$\PREPAREC$ to be a unitary satisfying
\[
\PREPAREC\ket{0} = \frac{1}{\sqrt{\sum_j \alpha_j^2}}\sum_j \alpha_j\ket{j}.
\]
and $\SELECTC = \sum_{j=0}^m \ket{j}\bra{j}\otimes \mathcal{B}[A_j]$,
the channel-LCU construction implements $\mathcal{E}\rbra{\rho} =\sum_j A_j\rho A_j^{\dagger}$
using the circuit shown in~\Cref{fig:channel-lcu}.




 
\subsection{Discussion}

\paragraph{Why choose the Kraus representation?}

A natural question is why ChannelIR chooses Kraus form as its core channel representation, rather than other standard representations such as circuit descriptions, Choi matrices, or Stinespring/unitary dilations. Our choice is driven by the needs of compilation. The IR should preserve channel semantics at a manipulable level, support algebraic optimization, and lower cleanly to the block-encoding and channel-LCU backend used in our implementation. Under these criteria, Kraus form provides the best balance (see~\Cref{tab:channel-representations} for the comparison).

\begin{table}[h]
\centering
\footnotesize
\caption{Comparison of channel representations with respect to ChannelIR design goals.}
\label{tab:channel-representations}
\vspace{-1em}
\begin{tabular}{p{2.3cm} p{2.4cm} p{3.4cm} p{3.1cm}}
\toprule
Representation 
& IR manipulation 
& Lowering to LCU 
& Representation cost \\
\midrule

Kraus
& Supports algebraic rewrite rules
& Directly maps to block-encoding and channel-LCU
& List of $n$-qubit operators \\

Choi matrix
& Hard to manipulate structurally
& Requires conversion to Kraus form
& A $2^{2n}\!\times\!2^{2n}$ matrix \\

Stinespring 
& Operator structure hidden in unitary
& Not aligned with channel-LCU synthesis
& Requires additional workspace \\

\bottomrule
\end{tabular}
\vspace{-1em}
\end{table}

\paragraph{Comparison with Cobble}
Cobble~\cite{yuan2025cobblecompilingblockencodings} is closely related to our \ChannelIR{} in that it elevates block-encodings to a first-class quantum programming abstraction, with a globally designed language for block-encodings of matrix operators, whose semantics are given by their circuit realizations. 

Cobble and \ChannelIR{} differ primarily in abstraction level and compilation focus. Cobble is a high-level language centered on composing block-encoded matrices, where block-encoding is the core abstraction and programs are organized around matrix algebra. In contrast, \ChannelIR{} is a multi-level intermediate representation tailored for quantum channel compilation: it includes Pauli primitives and Kraus operators as first-class objects, preserves channel structure for algebraic rewriting, and supports more flexible (including complex-coefficient) LCU constructions. This design enables \ChannelIR{} to eliminate redundancies and optimize at the channel level before circuit implementation, making it better suited for channel-first compilation workflows, whereas Cobble is stronger for expressing high-level polynomial operations over block-encoded matrices.

\begin{table}[ht]
\centering
\small
\caption{Comparison between Cobble and \ChannelIR.}
\label{tab:cobble-comparison}
\vspace{-0.5em}
\begin{tabular}{p{2.8cm}p{5.0cm}p{5.0cm}}
\toprule
\textbf{Aspect} & \textbf{Cobble} & \ChannelIR{} \\
\midrule
Design level
& High-level language for algebra over block-encoded matrices
& Multi-level intermediate representation for channel compilation \\

Atomic abstraction
& Block-encoding is the basic programming object
& Pauli primitives are basic objects; block-encoding is a lowering mechanism or optional primitive \\

Semantic focus
& Programs are organized around block-encoded matrix expressions
& Terms denote the underlying Kraus operators for the channel being compiled \\

Compilation goal
& Express and compose block-encoded linear-algebra operators
& Preserve Kraus and Pauli-sum structure for channel-first rewriting and circuit implementation \\

LCU coefficients
& Real-coefficient LCU
& Supports complex coefficients in LCU which may simplify circuit implementation (see Example \ref{example:tfim} and \ref{example:all-pauli}) \\


Strength
& Support polynomials over matrices being block-encoded
& Natural for channel compilation, Kraus-form simplification, and channel-LCU lowering \\
\bottomrule
\end{tabular}
\vspace{-1em}
\end{table}

\section{Structure-aware Optimization for LCU}
The previous section presents a circuit implementation path for channels in \ChannelIR{},
based on two layers of LCU constructions: channel-LCU at the level of Kraus operators,
and ordinary LCU within the block-encoding of each Kraus operator.
Although this implementation path is systematic, its naive realization introduces
heavily controlled multiplexor structures that dominate the final circuit cost.
The main task of this section is therefore to optimize the resulting \SELECT{} oracles,
which arise at different levels of the implementation pipeline and exhibit different
structural properties.

\subsection{Optimization Targets and Cost Model}
The optimization targets are the \SELECT{} oracles arising from channel-LCU and from
the block-encoding of individual Kraus operators. In both cases, the corresponding
multiplexor-style implementations require cascaded multi-controlled gates, which constitute
the main bottleneck in the final circuit cost.

Although both layers incur substantial overhead from multiplexor-style control,
they expose different amounts of structure to the compiler. The channel-LCU
\SELECT{} operates over block-encodings of Kraus operators and therefore treats
them largely as black-box components. By contrast, the Kraus-level
\SELECT{} ranges over Pauli operators preserved in the Kraus-level IR, exposing
additional algebraic structure that enables more specialized optimization.

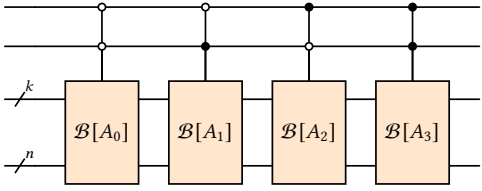
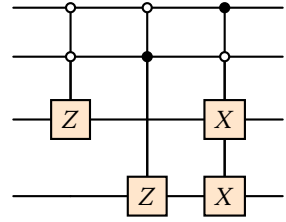
\begin{figure}[h]
\centering
\begin{subfigure}{0.58\linewidth}
\resizebox{0.8\textwidth}{!}{
    \begin{quantikz}
 &&\octrl{2} & \octrl{2} & \ctrl{2} & \ctrl{2} & \\
 && \octrl{1}& \ctrl{1} &\octrl{1}& \ctrl{1} &  \\
   &\qwbundle{k} & \gate[2, style={fill=orange!20}]{\mathcal{B}[A_0]} & \gate[2, style={fill=orange!20}]{\mathcal{B}[A_1]} & \gate[2, style={fill=orange!20}]{\mathcal{B}[A_2]} & \gate[2, style={fill=orange!20}]{\mathcal{B}[A_3]} & \\
& \qwbundle{n} &&&&& \\
\end{quantikz}}
\caption{Channel-level \SELECTC{} for a quantum channel
$\mathcal{M}(\cdot) = \sum_{j=0}^3 A_j \cdot A_j^{\dag}$, shown as a uniformly
controlled multiplexor over Kraus block-encodings.
}
\end{subfigure}
\hspace{0.02\linewidth}
\begin{subfigure}{0.37\linewidth}
\centering
\begin{quantikz}
 &\octrl{2} & \octrl{3} & \ctrl{3} &  \\
 & \octrl{1}& \ctrl{2} &\octrl{2}&   \\
    & \gate[1, style={fill=orange!20}]{Z} & & \gate[1, style={fill=orange!20}]{X} &  \\
  & &\gate[1, style={fill=orange!20}]{Z}&\gate[1, style={fill=orange!20}]{X}&  \\
\end{quantikz}
\caption{Kraus-level \SELECT{} for $\mathcal{B}[A_0]$, where
$A_0 = a_0 ZI + a_1 IZ + a_2 XX$.}
\end{subfigure}
\caption{Two optimization targets in the LCU implementation pipeline:
(a) the channel-level multiplexor over Kraus block-encodings, and
(b) the block-encoding-level multiplexor over Pauli terms.}
\label{fig:channel-be2}
\end{figure}

\paragraph{Primary Resource Metric.}
The primary resource metric in this section is the total control load of the multi-controlled gates induced by the two \SELECT{} oracles. 
This choice is justified by the fact that a multi-controlled gate with $n \geq 2$ controls requires at least $n + 1$ $T$ gates to realize~\cite{Beverland2020lower},
which is a key bottleneck in fault-tolerant quantum computing under prevailing quantum error-correction codes~\cite{Fowler2012surface}.
Both techniques target this bottleneck: Technique~I reduces control arity directly by flattening,
while Technique~II uses the weighted control cost $\mathrm{w}(l)\mathrm{w}_s(g_l)$ to capture the number of controlled Pauli operations induced by each decomposition term.

\subsection{Technique I: Conditional Flattening for Channel-Level \SELECT}
At the channel level, the \SELECT{} oracle is a uniformly controlled gate sequence over Kraus block-encodings, indexed by the Kraus register. 
Because the selected objects are already block-encoding circuits and are therefore treated as black-box components, this layer is more naturally 
amenable to generic control-flow optimization that reduces the arity of the induced multi-controlled gates.

We adopt conditional flattening~\cite{Yuan2024tcomplexity}, also known as unary
iteration~\cite{Babbush2018encoding}, as our first optimization technique. Its
core idea is to replace a uniformly controlled family of high-control gates by a
sequence of gates with smaller control arity, using ancilla-assisted updates of
the control pattern. Concretely, a multi-controlled gate with $n$ controls is
replaced by one with $n-1$ controls together with two Toffoli gates, at the cost
of one ancilla qubit. Applied to a channel-LCU object with $N$ Kraus operators,
this transformation reduces the $T$-count of the channel-level \SELECT{} from
$\mathcal{O}(N\log N)$ to $4N-4$~\cite{Babbush2018encoding}, while requiring
$\mathcal{O}(\log N)$ additional ancilla qubits. We therefore use conditional
flattening as a generic optimization for channel-level \SELECT{}, especially in
settings where reducing $T$-gate cost is more important than minimizing ancilla
usage.
\begin{figure}
\begin{minipage}{0.19\textwidth}
\resizebox{0.99\textwidth}{!}{
\begin{quantikz}
\lstick{$c_0$} &   \ctrl{2} & & \ctrl{2} & \\
\lstick{$c_1$} &  \ctrl{1} & & \ctrl{1} & \\
\lstick{$\ket{0}$} & \targ{} & \ctrl{2} & \targ{} &\rstick{$\ket{0}$} \\
\lstick{$c_2$} & &   \ctrl{1} & &  \\
\lstick{$\ket{\psi}$} & \qwbundle{n}  & \gate{\mathcal{B}[A_0]} &  & \\
\end{quantikz}
}
\caption{Flattening a $3$-controlled gate.}
\label{fig:flattening}
\end{minipage}
\hspace{0.02\textwidth}
\begin{minipage}{0.77\textwidth}
\resizebox{\textwidth}{!}{
\centering
\begin{quantikz}[column sep = 0.2cm, row sep = 0.5cm]
\lstick{$c_0$} & \octrl{1} &&&&&&&&&&&&& &&&&&&&&&&& &&&& \octrl{1} &&\\
&\wireoverride{n} & \ctrl{2} & & & & & & \ctrl{2}& & & & & & \ctrl{2} & \targ{}  &\ctrl{2} & & & & & & \ctrl{2}& & & & & & \ctrl{2} & &\wireoverride{}  \\
\lstick{$c_1$}&  & \octrl{1} & & & & & & &  & & & & & \ctrl{1} & & \octrl{1} & & & & & & &  & & & & & \ctrl{1} & & &  \\
&\wireoverride{n}& \wireoverride{n} & \ctrl{2} &  & \ctrl{2} & & \ctrl{2} & \targ{} &\ctrl{2} & & \ctrl{2} & & \ctrl{2} & & \wireoverride{n} &\wireoverride{n}& \ctrl{2} &  & \ctrl{2} & & \ctrl{2} & \targ{} &\ctrl{2} & & \ctrl{2} & & \ctrl{2} & & \wireoverride{n}   \\
\lstick{$c_2$} & & & \octrl{1} & & & & \ctrl{1} & & \octrl{1} & & & & \ctrl{1} &&& &  \octrl{1} & & & & \ctrl{1} & & \octrl{1} & & & & \ctrl{1} & &&&\\
& \wireoverride{n} & \wireoverride{n} & \wireoverride{n} & \ctrl{1}& \targ{} & \ctrl{1} & & \wireoverride{n} & \wireoverride{n} & \ctrl{1} & \targ{} & \ctrl{1} & & \wireoverride{n} & \wireoverride{n} & \wireoverride{n} & \wireoverride{n} & \ctrl{1}& \targ{} & \ctrl{1} & & \wireoverride{n} & \wireoverride{n} & \ctrl{1} & \targ{} & \ctrl{1} & & \wireoverride{n} & \wireoverride{n}\\
\lstick{$\ket{\psi}$} & \qwbundle{n} & & & \gate{\mathcal{B}[A_0]} & &  \gate{\mathcal{B}[A_1]} & & & &\gate{\mathcal{B}[A_2]} & & \gate{\mathcal{B}[A_3]} & & & & & &  \gate{\mathcal{B}[A_4]} & &  \gate{\mathcal{B}[A_5]} & & & &\gate{\mathcal{B}[A_6]} & & \gate{\mathcal{B}[A_7]} & & &&  \\
\end{quantikz}
}
\vspace{-1em}
\caption{A channel-level \SELECTC{} oracle with $8$ Kraus operators after
optimization by conditional flattening (unary iteration).}
\label{fig:optI-circuit}
\end{minipage}
\end{figure}


\subsection{Technique II: Structure-Aware Heuristic for Pauli-Sum \SELECT}
We now turn to the block-encoding-level \SELECT{}, where the Pauli-sum representation exposes 
additional algebraic structure and enables a more specialized heuristic optimization.

\subsubsection{Overview and Optimization Problem}

Unlike Technique~I, the block-encoding-level \SELECT{} oracle is not built over opaque block-encoding components. 
Because \ChannelIR{} preserves each Kraus operator as a Pauli-sum expression, the compiler still has access to Pauli-level algebraic structure when constructing the inner multiplexor. 
Technique~II exploits this structure by optimizing the assignment of Pauli strings to control addresses and the factorization of the resulting multiplexor into simpler controlled operations.

We call the latter factorization a monotone-control decomposition: each factor is triggered by the subset of control bits set to $1$, rather than by an exact control string. 
Accordingly, the optimization is formulated over Pauli-string structure alone. Phase information is tracked consistently with the chosen assignment but does not affect the objective, and unused control addresses are harmless because their amplitudes in the \PREPARE{} oracle are zero. We thus formulate the following problem.
\begin{problem}[Optimization of Pauli-sum \SELECT{}]
Given an $n$-qubit Pauli sum
\(
A = \sum_{j=1}^{m} \beta_j P_j,
\)
where each $P_j$ is a Pauli string and $s = \lceil \log_2 m \rceil$, find

(1) an injective assignment $\pi : \{1,\dots,m\} \to \{0,1\}^s$ of Pauli strings to control addresses, and

(2) a family of Pauli strings $\{g_l\}_{l \in \{0,1\}^s}$ forming a monotone-control decomposition,

such that for every assigned address $\pi(j)$,
\(
P_j = \prod_{l \leq_b \pi(j)} g_l,
\)
where $\leq_b$ denotes the bitwise partial order on control addresses.

The objective is to minimize the weighted control cost
\(
C = \sum_{l \in \{0,1\}^s} \mathrm{w}(l)\,\mathrm{w}_s(g_l),
\)
where $\mathrm{w}(l)$ is the Hamming weight of the control address and $\mathrm{w}_s(g_l)$ is the Pauli weight of $g_l$.
\end{problem}

This isolates the combinatorial core of the optimization problem. Since solving it exactly is intractable in general, we develop a structure-aware heuristic specialized to Pauli sums.

\subsubsection{Semantic Invariants and Pauli Reformulation}
This heuristic relies on three ingredients: semantic invariance under permutation, existence of monotone-control decompositions, and an efficient Pauli-specific reformulation. We introduce them in turn.


We first observe that permuting the terms of an LCU does not change the postselected operator implemented by its block-encoding circuit, provided that the coefficients are permuted accordingly. This semantic invariance justifies searching for a cost-improving address assignment.
\begin{proposition}[Equivalence of LCU Block-encoding Circuits Under Permutation]\label{prop:permute}
Suppose we have two LCU expressions $ C_1 $ and $ C_2 $ defined as:
$$
C_1 = \sum_{k=1}^m \beta_k V_k \quad \text{and} \quad C_2 = \sum_{k=1}^m \beta_{\pi(k)} V_{\pi(k)}
$$
where $\pi$ is a permutation of $\{1,2,\dots,m\}$. Then we have the following equivalence: 
$$
\forall \rho \in \mathcal{D}(\mathcal{H}), (\bra{0}\otimes I )\mathcal{B}[C_1] (\ket{0}\bra{0}\otimes\rho)(\ket{0}\otimes I) = (\bra{0}\otimes I )\mathcal{B}[C_2] (\ket{0}\bra{0}\otimes\rho)(\ket{0}\otimes I).
$$
\end{proposition}
Next, we show that once a control-address assignment is fixed, a corresponding monotone-control decomposition always exists. 
This does not yet yield an efficient optimizer, but it ensures that the search space is well defined; for convenience, we state the existence result first for arbitrary unitary matrices.

\begin{proposition}[Existence of monotone-control decomposition]\label{prop:existence}
For an arbitrary set of $n$ unitary matrices $\{V_1, \cdots, V_n\}$, there exists a monotone-control decomposition $\{U_1, \cdots, U_n\}$ such that 
\(
\forall k \in \{1, \cdots, n\}, V_k = \Pi_{(l\leq_b k)} U_l.
\)

\end{proposition}
The above existence result is too weak for implementation purposes at the level of arbitrary unitary matrices. 
The key simplification in our setting is that the objects being manipulated are Pauli strings, which admit a compact binary representation and support multiplication by bitwise operations rather than dense matrix algebra.

\begin{definition}[Symplectic representations of Pauli strings~\cite{nielsen2010quantum}]
Suppose $P =\ii^d g$ is an $n$-qubit Pauli string, where $g \in \{ I, X,Y,Z\}^{\otimes n}, d\in\{0, 1, 2,3\}$. Then $r(P) = r(g) = (g_z\lvert g_x) \in \{0,1\}^{\otimes 2n}$ is the symplectic vector of $g$. For all $i \in \{0, \cdots, n-1\}$, 
\[
r(g)_i = \begin{cases} 
1 & \text{if } g_i = X \text{ or } Y \\
0 & \text{if } g_i = I \text{ or } Z 
\end{cases}
\quad
r(g)_{n+i} = \begin{cases} 
1 & \text{if } g_i = Y \text{ or } Z \\
0 & \text{if } g_i = I \text{ or } X 
\end{cases}
\]
For example, the $5$-qubit Pauli string $g = XIZYI$ has the symplectic vector $r(g) = (10010\lvert00110)$.  
\end{definition}
Under this representation, multiplication of Pauli strings reduces to bitwise vector addition together with a simple phase update. As a result, both address assignment and decomposition recovery can be carried out over binary data rather than dense matrices.

\subsubsection{A Heuristic Algorithm}

The optimization problem above is combinatorial and appears intractable in general, so we use a structure-aware heuristic specialized to Pauli sums. At a high level, the heuristic proceeds in two stages: it first assigns Pauli strings to control addresses, and then reconstructs a monotone-control decomposition consistent with that assignment. Thus, the input is a Pauli-sum LCU object, while the output is an equivalent implementation plan---an address assignment together with decomposition factors---that preserves the implemented postselected operator and heuristically minimizes the weighted control cost. Figure~\ref{fig:technique2-flow} summarizes this workflow.

In the first stage, the algorithm searches for a favorable placement of Pauli strings in the control space. Rather than searching over all permutations, it uses the maximum-coverage solution as a starting point, enlarges the search space to nearby candidate subspaces, and finally selects the address assignment minimizing the weighted control cost. The output of this stage is an ordered address table for the target Pauli strings.

In the second stage, the algorithm reconstructs the monotone-control decomposition from the chosen address assignment. Once the target Pauli strings have been placed, the factors $g_l$ are recovered by inversion along the bitwise partial order. This stage also determines the corresponding phase information, yielding a decomposition that is semantically consistent with the original Pauli-sum object.

For readability, the main text presents only the high-level workflow and the role of each stage. Full pseudocode and implementation details, including sorting, binary conversion, subspace assignment, and inversion procedures, are deferred to the appendix.

\begin{figure}[t]
\centering
\resizebox{0.9\textwidth}{!}{
\begin{tikzpicture}[
>=Latex,
node distance=6mm and 7mm,
font=\small,
box1/.style={
draw,
rounded corners=2pt,
thick,
fill=orange!12,
minimum height=9mm,
minimum width=2.9cm,
align=center,
inner sep=3pt
},
box2/.style={
draw,
rounded corners=2pt,
thick,
fill=blue!10,
minimum height=9mm,
minimum width=2.9cm,
align=center,
inner sep=3pt
},
bridge/.style={
draw,
rounded corners=2pt,
thick,
fill=gray!10,
minimum height=8mm,
minimum width=2.6cm,
align=center,
inner sep=3pt
},
labelbox/.style={
draw=none,
fill=none,
font=\bfseries\small,
align=center
},
arr/.style={->, thick}
]

\node[box1, xshift = -5mm] (symp) {Convert Pauli strings\\to symplectic form};
\node[labelbox, left= 5mm of symp] (s1label) {Stage I\\Address Assignment};
\node[box1, right=of symp] (subspace) {Generate candidate\\covering subspaces};
\node[box1, right=of subspace] (assign) {Select the best \\ address assignment};

\node[bridge, below=3mm of assign] (table) {Ordered address table};

\node[box2, below=15mm of assign] (invert) {Recover monotone-\\control factors};
\node[box2, left=of invert] (phase) {Reconstruct\\phase data};
\node[box2, left=of phase] (output) {Output implementation\\plan};
\node[labelbox, left = 6mm of output] (s2label) {Stage II\\Reconstruction};
\draw[arr] (symp) -- (subspace);
\draw[arr] (subspace) -- (assign);
\draw[arr] (assign) -- (table);
\draw[arr] (table) -- (invert);
\draw[arr] (invert) -- (phase);
\draw[arr] (phase) -- (output);
\node[draw, dashed, rounded corners=4pt, inner sep=5pt, fit=(symp)(subspace)(assign), label=above:{\bfseries\small structure discovery}] {};
\node[draw, dashed, rounded corners=4pt, inner sep=5pt, fit=(invert)(phase)(output), label=above:{\bfseries\small semantic reconstruction}] {};
\end{tikzpicture}
}
\caption{A sketched workflow of Technique~II. The first stage assigns Pauli strings into low-cost control addresses by exploiting symplectic structure, while the second stage reconstructs a monotone-control decomposition and the associated phase information for the final implementation.}
\label{fig:technique2-flow}
\end{figure}
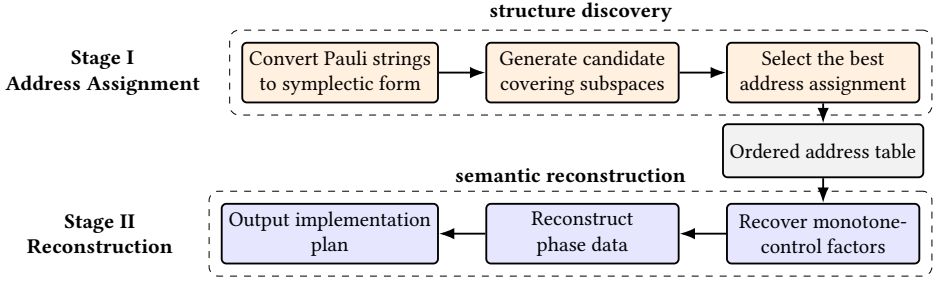
\subsubsection{Examples and Discussion}
We present two examples that illustrate the effectiveness of Technique~II. 
The first demonstrates how the heuristic simplifies the block-encoding of a representative physical instance, 
while the second shows that the method attains optimality on a highly structured family of Pauli-sum objects.
\begin{example}
\label{example:tfim}
Consider the $1D$ Transverse-Field Ising model (TFIM) with damping for $N = 3$ qubits: 
\[
H = -\rbra*{\sum_{i=1}^{N-1} Z_iZ_{i+1} + Z_NZ_1} - \sum_{i=1}^{N-1}X_i, \quad L_j = \sqrt{\gamma}(X_j - \ii Y_j)/2, j \in \{1, \cdots, N\}
\]
The optimization target is the block-encoding circuit of Kraus operator $A_0$ created by the channel-LCU method~\cite{cleve2017lindblad}: 
\[
A_0 = \rbra*{1 - \frac{N\delta \gamma}{2}}I + \ii\delta\rbra*{\sum_{i=1}^{N-1} Z_iZ_{i+1} + Z_NZ_1+\sum_{i=1}^{N}X_i} - \delta \gamma \sum_{i=1}^{N}Z_i
\]
There are $10$ terms in total, so $4$ control qubits are required. Table~\ref{tab:example-TFIM} lists the assigned control values, the reordered target Pauli terms, and the resulting decomposition terms $g$.
For this instance, the assignment is optimal under our cost metric, yielding an implementation with $3$ CZ gates and $3$ Toffoli gates.
This substantially improves over the original construction, which uses $9$ multi-controlled gates, each with $4$ control qubits.
\begin{table}[h]
\centering
\caption{The optimal order of Pauli strings and corresponding $g$s and coefficients in TFIM damping model.}
\vspace{-0.5em}
\resizebox{0.9\linewidth}{!}{
\begin{tabular}{|c|c|c|c|c|c|c|c|c|c|c|}
\hline
\textbf{control value} & $0000$  & $0001$  & $0010$  &  $0011$ &  $0100$ & $0101$  &  $0110$ &  $1001$ &  $1010$ & $1100$ \\
\hline
\textbf{operator $P$} & $I$ & $-Z_1$ & $-Z_2$  & $\ii Z_1Z_2$ &  $-Z_3$& $\ii Z_1Z_3$ & $\ii Z_2Z_3$& $\ii X_1$  & $\ii X_2$ & $\ii X_3$  \\
\hline
\textbf{operator $g$} & $I$ & $-Z_1$ & $-Z_2$ & $ I$ & $-Z_3$& $I$ & $ I$& $Y_1$  & $Y_2$ & $Y_3$ \\
\hline
\textbf{coefficient} &$(1 - \frac{N\delta \gamma}{2})$  & $\delta\gamma$ & $\delta\gamma$ & $\ii\delta$ & $\delta\gamma$ & $\ii\delta$ &$\ii\delta$  &$\delta$  & $\delta$ &  $\delta$\\
\hline
\end{tabular}
}
\label{tab:example-TFIM}
\end{table}
\vspace{-1em}
\end{example} 

\begin{example}
This example shows that, for a highly structured Pauli family, Technique~II specializes to a closed-form optimal pattern. 
Consider the linear combination of all $n$-qubit Pauli operators modulo phase: 
$A = \sum_{P_k\in \{I,X,Y,Z\}^{\otimes n}} \beta_k P_k$. The original encoding circuit requires $4^n$ multi-controlled gates, each with $2n$ control qubits.
Our method successfully finds an optimal activated value decomposition for this case: 
\[
\begin{cases}
g_k  = X_{d+1}, \ 2^0 \leq k =2^d \leq 2^{n-1} \ \\

g_k  = Z_{d+1}, \  2^n \leq k =2^{d+n} \leq 2^{2n-1} \\
g_k = I, k = \text{other} \\
\end{cases}
\]
This family is sufficiently structured that the decomposition can be written down in closed form, 
without any additional search for a cost-improving permutation. 
In particular, when the Pauli terms are arranged in reversed lexicographical order, i.e. $k = \mathrm{rev}(r(P_k))$, the above construction is already optimal under the weighted-control cost metric. 
Moreover, because complex coefficients can be absorbed into the \PREPARE{} routine of the LCU construction, no additional controlled-phase gates are required in the resulting \SELECT{} oracle. 
Consequently, the original implementation is compressed to an implementation using only $\mathbf{2n}$ single-controlled gates.
\label{example:all-pauli}
\end{example}

Example \ref{example:tfim} shows effectiveness on a realistic sparse Pauli-sum instance, while Example \ref{example:all-pauli} demonstrates that on a highly structured family
the method admits a closed-form optimal pattern. Together, these examples show two complementary regimes of Technique~II: 
on realistic sparse instances it provides an effective heuristic simplification, while on highly structured Pauli families it admits a closed-form optimum. More broadly, they suggest that the benefit of Technique~II grows with the amount of Pauli-level algebraic structure preserved in the implementation.

\section{\LindFront{}: A Lindbladian Frontend to \ChannelIR}

Lindbladian generators provide a natural language description for many open-system quantum algorithms, while our compilation operates on discrete quantum channels expressed in \ChannelIR. To connect these two levels of abstraction, we introduce \LindFront, a frontend that transforms Lindbladian generator descriptions into short-time channel representations that can be compiled by the rest of our compilation framework. 



\subsection{Frontend Interface}
\LindFront{} accepts a Lindbladian generator specified in the standard form consisting of a Hamiltonian term and a collection of jump operators describing dissipative dynamics. 
In practice, these operators can be provided either as explicit matrices or as Pauli-sum expressions. 

For operators provided 
as explicit matrices, the first step is to convert all Lindbladian operators into matrix-sum representation over Pauli strings. 
This representation matches the syntax of \ChannelIR{} and ensures compatibility with later compilation stages.


The output superoperator is expressed in \ChannelIR{} in Kraus form with operators represented as Pauli matrix sums, preserving algebraic structure that can later be implemented and optimized. 

\subsection{Lowering Algorithms and Guarantees}

\LindFront{} currently supports two such strategies based on prior work~\cite{cleve2017lindblad,Li2023simulating}. 
These strategies differ in their approximation accuracy but share the same input-output interface.

\textit{First-order short-time construction.}
The simplest lowering strategy follows the first-order short-time expansion proposed in~\cite{cleve2017lindblad}. Given a Lindbladian generator consisting of a Hamiltonian term $H$ and jump operators $\{L_j\}_{j=1}^m$, the evolution over a small timestep $\delta$ can be approximated by a superoperator $\mathcal{M}_{\delta}\rbra{\rho} = \sum_{j=0}^m A_j \rho A_j^{\dagger}$
with 
\(
A_0 = I - \frac{\delta}{2}\sum_j L_j^\dagger L_j - \ii \delta H,
\
A_j = \sqrt{\delta}\, L_j .
\)
This construction directly produces a Kraus representation of a superoperator approximating $e^{\delta\mathcal{L}}$. 



In our compiler, these Kraus operators are represented as \textsc{ChannelIR} terms. The resulting superoperator can then be optimized and synthesized by the generic channel compilation pipeline described earlier.

\textit{Higher-order expansion.}
While the first-order construction is simple and efficient, achieving high simulation accuracy may require very small timesteps. To address this limitation, \LindFront{} also supports a higher-order lowering strategy based on the series-expansion technique of~\cite{Li2023simulating}. Utilizing this method, \LindFront{} constructs a more accurate superoperator as an approximation which is expressed in Kraus form within \ChannelIR. Because both strategies produce superoperators of the same structural form, they can both be given as input into the compiler for later processing without modifications.

\textit{Guarantees.}
The frontend provides guarantees at two levels. First, regardless of the lowering strategy used, the frontend always outputs a well-typed \ChannelIR{} program. This ensures that the resulting program can be safely processed by the circuit implementation and optimization passes in our framework. Second, the accuracy of the generated channel as an approximation of $e^{\delta\mathcal{L}}$ depends on the selected lowering method. The first-order construction of~\cite{cleve2017lindblad} introduces an $O(\delta^2)$ approximation error, while higher-order constructions such as~\cite{Li2023simulating} can achieve improved accuracy per timestep. These approximation guarantees are implied directly from the correctness of the underlying algorithms~\cite{cleve2017lindblad,Li2023simulating}.
\section{Evaluation}

We implemented our compilation framework as a Python software package using the Qiskit SDK~\cite{qiskit2024}. The package comprises the channel compiler, the block-encoding optimizer, and the explicit Lindbladian compilation frontend. We organize the evaluation around the following research questions:

\textit{RQ1}. How much does the channel-first compilation pipeline reduce the cost of end-to-end channel simulation circuits, and how does it compare with a circuit-first baseline?

\textit{RQ2}. How much do the two optimization techniques contribute individually and jointly, and how do they compare with existing control-flow optimization baselines?

\textit{RQ3}. Does the Lindbladian frontend produce correct circuits, and what evidence do we have for compilation scalability?

All experiments were conducted on a MacBook Air equipped with an Apple M4 CPU and $16$GB of RAM.

\subsection{RQ1: End-to-End Benefits of Channel-First Compilation}

We first evaluate the cost of complete channel-simulation circuits generated from \textsc{ChannelIR}. This experiment addresses the central question of whether preserving channel and Kraus structure throughout the compilation pipeline yields lower-cost final circuits than a circuit-first realization.

We use two benchmark families: the TFIM damping model generated by the Lindbladian frontend and the hypercube random-walk channel~\cite{TW25}. 
The four compilation settings correspond to enabling optimization technique I on the channel-level selection circuit (Flat), optimization technique II within Kraus block-encodings (Order), or both. 
Where available, we also include Qiskit's Stinespring decomposition as a circuit-first baseline.
Table~\ref{tab:hypercube_lcu_counts_ratios} reports the resource comparison, with `Basic+Basic' denoting the unoptimized channel-first pipeline and all ratios measured against this baseline. 

The results show that both optimization techniques substantially reduce the final circuit cost, with the combined setting achieving up to a $99\%$ reduction in gate count relative to the unoptimized baseline. Individually, the two techniques reduce gate counts to $3\%$--$56\%$ of the baseline, while their combination reduces all reported metrics to $0.75\%$--$5.7\%$. This advantage becomes more pronounced as the problem size increases. Moreover, the channel-first approach consistently outperforms the Stinespring dilation: at $N = 8$, even the unoptimized channel-first baseline requires only $4\%$ of the gates and $2\%$ of the compilation time, whereas at $N = 12$ the Stinespring approach fails to return within $10$ minutes while channel-LCU continues to scale efficiently.
\begin{table*}[t]
\centering
\caption{Resource comparison for channel-LCU circuits on various benchmarks after Qiskit's transpilation. Ratios are measured against the `Basic+Basic' baseline.}
\label{tab:hypercube_lcu_counts_ratios}
\vspace{-0.5em}
\resizebox{\textwidth}{!}{
\begin{threeparttable}
\begin{tabular}{c l c c cccc cccc}
\toprule
\multirow{2}{*}{Benchmark} & \multirow{2}{*}{Case} & \multirow{2}{*}{$\#$Kraus} & \multirow{2}{*}{$\#$Pauli Terms} & \multicolumn{4}{c}{Count} & \multicolumn{4}{c}{Ratio vs. Basic+Basic ($\%$)} \\
\cmidrule(lr){5-8}\cmidrule(lr){9-12}
& & & & Depth & $\#$Gates & $\#$CX & $\#$U3 & Depth & $\#$Gates & $\#$CX & $\#$U3 \\
\midrule

\multirow{4}{*}{TFIM-8}
& Basic+Basic  & 9 & 41 & 519,398 & 737,823 & 358,127 & 379,696 & $100.00$ & $100.00$ & $100.00$ & $100.00$ \\
& Flat+Basic   & 9 & 41 &  55,407 &  69,124 &  29,610 &  39,492 &  $10.67$ &   $9.37$ &   $8.27$ &  $10.40$ \\
& Basic+Order  & 9 & 41 & 292,064 & 415,943 & 202,023 & 213,920 &  $56.23$ &  $56.37$ &  $56.41$ &  $56.34$ \\
& Flat+Order   & 9 & 41 &  29,984 &  37,850 &  16,080 &  21,748 &   $5.77$ &   $5.13$ &   $4.49$ &   $5.73$ \\
\midrule
\multirow{4}{*}{TFIM-12}
& Basic+Basic  & 13 & 61 & 1,294,267 & 1,864,327 & 905,065 & 959,262 & $100.00$ & $100.00$ & $100.00$ & $100.00$ \\
& Flat+Basic   & 13 & 61 &   143,367 &   178,196 &  76,597 & 101,571 &  $11.08$ &   $9.56$ &   $8.46$ &  $10.59$ \\
& Basic+Order  & 13 & 61 &   148,957 &   215,371 & 104,767 & 110,604 &  $11.51$ &  $11.55$ &  $11.58$ &  $11.53$ \\
& Flat+Order   & 13 & 61 &    16,086 &    20,227 &   8,717 &  11,482 &   $1.24$ &   $1.08$ &   $0.96$ &   $1.20$ \\
\midrule
\multirow{5}{*}{hypcubeW-8}
& Stinespring~\cite{qiskit2024} & 16 & 64 &  -- & 2,164,787 & 1,081,133 & 1,083,654 & -- & $2570.39$ & $2648.80$ & $2496.67$ \\
& Basic+Basic & 16  & 64  & 60,268  & 84,220  & 40,816  & 43,404  & $100.00$ & $100.00$ & $100.00$ & $100.00$ \\
& Flat+Basic  & 16  & 64  & 5,257   & 6,766   & 2,816   & 3,920   & $8.72$ & $8.03$ & $6.90$ & $9.03$ \\
& Basic+Order & 16  & 64  & 12,686  & 18,340   & 8,918   & 9,422   & $21.05$ & $21.78$ & $21.85$ & $21.71$ \\
& Flat+Order  & 16  & 64  & 1,177   & 1,709     & 712     & 967     & $1.95$ & $2.03$ & $1.74$ & $2.23$\\
\midrule
\multirow{5}{*}{hypcubeW-12}
& Stinespring~\cite{qiskit2024} & 24 & 96 & N/A \tnote{1} & N/A & N/A & N/A & -- & -- & -- & -- \\
& Basic+Basic     & 24 & 96  & 147,680 & 201,770 & 98,588  & 103,182 & $100.00$ & $100.00$ & $100.00$ & $100.00$ \\
& Flat+Basic       & 24 & 96  & 7,919   & 10,213  & 4,249   & 5,916   & $5.36$& $5.06$ & $4.31$ & $5.73$ \\
& Basic+Order & 24 & 96  & 31,461  & 45,145  & 22,096  & 23,049  & $21.30$ & $22.37$ & $22.41$ & $22.34$ \\
& Flat+Order   & 24 & 96  & 1,799   & 2,628   & 1,093   & 1,487   & $1.22$ & $1.30$ & $1.11$ & $1.44$ \\
\midrule
\multirow{4}{*}{hypcubeW-20}
& Basic+Basic   & 40 & 160 & 339,491 & 499,561 & 245,977 & 253,584 & $100.00$ & $100.00$ & $100.00$ & $100.00$ \\
& Flat+Basic    & 40 & 160 & 13,230  & 17,076  & 7,104   & 9,890   & $3.90$ & $3.42$ & $2.89$ & $3.90$ \\
& Basic+Order & 40 & 160 & 72,303  & 111,123 & 54,829  & 56,294  & $21.30$ & $22.24$ & $22.29$ & $22.20$ \\
& Flat+Order   & 40 & 160 & 3,030   & 4,435   & 1,844   & 2,509   & $0.89$ & $0.89$ & $0.75$ & $0.99$ \\
\midrule
\multirow{4}{*}{hypcubeW-28}
& Basic+Basic & 56 & 224 & 475,251 & 699,452 & 344,409 & 355,043 & $100.00$ & $100.00$ & $100.00$ & $100.00$ \\
& Flat+Basic  & 56 & 224 & 18,501  & 23,868  & 9,932   & 13,824  & $3.89$ & $3.41$ & $2.88$ & $3.89$ \\
& Basic+Order & 56 & 224 & 101,209 & 155,562 & 76,761  & 78,801  & $21.30$ & $22.24$ & $22.29$ & $22.19$ \\
& Flat+Order   & 56 & 224 & 4,221   & 6,171   & 2,568   & 3,491   & $0.89$ & $0.88$ & $0.75$ & $0.98$ \\
\bottomrule
\end{tabular}
\begin{tablenotes}
\item[1] The program was terminated before returning; therefore no results are reported.
\end{tablenotes}
\end{threeparttable}
}
\end{table*}

We further examine compilation latency by separating the cost of circuit construction from \ChannelIR{} and the cost of transpilation to a universal gate set. 
On the hypercube benchmarks, optimization improves not only the final circuit size but also the overall compilation workflow: 
lower-cost channel-LCU circuits are faster to construct and substantially easier to transpile, leading to clear reductions in total end-to-end compilation time. 
This effect becomes more pronounced as the benchmark size grows, indicating that the structural simplifications introduced by the channel-first pipeline benefit both the generated circuits and downstream compiler passes.
\begin{table*}[t]
\centering
\begin{minipage}{0.7\linewidth}
\centering

\setlength{\tabcolsep}{4pt}
\caption{Compilation times for the hypercube random-walk benchmark. Units are seconds.}
\label{tab:hypercube_compile_transpile}
\vspace{2ex}
\resizebox{.95\linewidth}{!}{
\begin{tabular}{l|cccc|cccc}
\toprule
\multirow{2}{*}{Method}
& \multicolumn{4}{c|}{Construct time (s)}
& \multicolumn{4}{c}{Transpile time (s)} \\
\cmidrule(lr){2-5}\cmidrule(lr){6-9}
& 
 $N=8$ & $N=12$  & $N=20$ & $N=28$ & 
 $N=8$ & $N=12$  & $N=20$ & $N=28$ \\
\midrule
Basic+Basic
 & 0.096 & 0.155   & 0.268 & 0.390
 & 6.034 & 19.044  & 60.697 & 88.230 \\
Flat+Basic
 & 0.125 & 0.195  & 0.337 & 0.498
 & 0.115 & 0.207  & 0.305 & 0.399 \\
Basic+Order
 & 0.074 & 0.118  & 0.214 & 0.301
 & 0.503 & 1.303  & 3.160 & 4.575 \\
Flat+Order
 & 0.100 & 0.152  & 0.258 & 0.372
 & 0.049 & 0.105  & 0.117 & 0.204 \\
Stinespring~\cite{qiskit2024}
 & 0.018 & 92.178  & $N/A$ & $N/A$
 & 291.160 & $N/A$  & $N/A$ & $N/A$ \\
\bottomrule
\end{tabular}
}
\end{minipage}%
\hspace{0.02\linewidth}
\begin{minipage}{0.256\linewidth}
\centering
\flushleft
\caption{Wall-clock simulation time for `mesolve' and our method.}
\resizebox{\linewidth}{!}{
        \begin{tabular}{lcc}
            \toprule
            $N$ & mesolve~\cite{qutip5} & Ours \\
            \midrule
            3 & $0.0013s$ & $0.09s$ \\
            5 & $0.30s$  & $0.22s$ \\
            7 & $904.49s$ & $0.89s$ \\
            8 & $N/A$  &  $3.03s$\\
            10 &  $N/A$ & $17.77s$ \\
            12 & $N/A$ & $188.03s$ \\
            \bottomrule
        \end{tabular}
        }
        \label{tab:sim_time_clcu}
\end{minipage}
\end{table*}

As complementary evidence of scalability, we also compare against QuTiP's `mesolve'~\cite{qutip5} on the TFIM benchmark. 
Although this comparison is not a phase-by-phase compiler measurement, it shows that our approach remains practical on instances for which the classical density-matrix solver already becomes prohibitively slow or fails to complete, suggesting that the compiled channel-based workflow scales more favorably in this regime.

\subsection{RQ2: Impact of the Two Optimization Techniques}

We next isolate the contribution of the optimization passes themselves. We first study block-encoding circuits in isolation, where the structural effect of monotone-control decomposition is most apparent. We then compare against an external control-flow optimization baseline to position our results relative to prior work.

\subsubsection{RQ2.1: Optimization Impact on Block-Encodings}

The first benchmark set isolates block-encoding circuits for Kraus operators expressed as Pauli sums. 
We use both TFIM-derived Kraus operators and random Pauli-sum instances, and compare three variants where available: 
the unoptimized baseline (Basic), conditional flattening (Flat) and monotone-control decomposition (Order).

\begin{table}[ht]
\centering
\caption{Metrics of block-encoding circuits from two LCU benchmarks.}
\resizebox{0.9\linewidth}{!}{
\begin{tabular}{c c l ccccc cccc}
\toprule
\multirow{2}{*}{Benchmark} & \multirow{2}{*}{\#Pauli terms} & \multirow{2}{*}{Method} & \multicolumn{5}{c}{Count} & \multicolumn{4}{c}{Ratio vs. No (\%)} \\
\cmidrule{4-8}\cmidrule{9-12}
& & & Qubits & Depth & $\#$Gates & $\#$CX & $\#$U3 & Depth & $\#$Gates & $\#$CX & $\#$U3 \\
\midrule
\multirow{3}{*}{TFIM-4-$A_0$}
& \multirow{3}{*}{13}
& Basic           & 8  & 986   & 1400  & 653  & 747  & 100.00 & 100.00 & 100.00 & 100.00 \\
& & Flat         & 12 & 311   & 521   & 207  & 286  & 31.54  & 37.21  & 31.70  & 38.29  \\
& & Order       & 9  & 54    & 88    & 32   & 56   & 5.48   & 6.29   & 4.90   & 7.50   \\
\midrule
\multirow{3}{*}{TFIM-8-$A_0$}
& \multirow{3}{*}{25}
& Basic           & 13 & 2944  & 4155  & 1973 & 2182 & 100.00 & 100.00 & 100.00 & 100.00 \\
& & Flat         & 18 & 630   & 1049  & 421  & 574  & 21.40  & 25.25  & 21.34  & 26.31  \\
& & Order      & 17 & 106   & 176   & 64   & 112  & 3.60   & 4.24   & 3.24   & 5.13   \\
\midrule
\multirow{3}{*}{TFIM-12-$A_0$}
& \multirow{3}{*}{37}
& Basic           & 18 & 7047  & 9595  & 4637 & 4958 & 100.00 & 100.00 & 100.00 & 100.00 \\
& & Flat         & 24 & 1008  & 1641  & 667  & 894  & 14.30  & 17.10  & 14.38  & 18.03  \\
& & Order       & 25 & 158   & 264   & 96   & 168  & 2.24   & 2.75   & 2.07   & 3.39   \\
\midrule
\multirow{3}{*}{RndPauli-4}
& \multirow{3}{*}{12}
& Basic     & 8  & 1870   & 2166   & 986    & 1180   & 100.00 & 100.00 & 100.00 & 100.00 \\
& & Flat   & 12 & 276    & 462    & 196    & 242    & 14.76  & 21.33  & 19.88  & 20.51  \\
& & Order  & 8  & 673    & 806    & 386    & 420    & 35.99  & 37.21  & 39.15  & 35.59  \\
\hline
\multirow{3}{*}{RndPauli-8}
& \multirow{3}{*}{120}
& Basic     & 15 & 66128  & 72404  & 34545  & 37859  & 100.00 & 100.00 & 100.00 & 100.00 \\
& & Flat   & 22 & 3163   & 5338   & 2352   & 2746   & 4.78   & 7.37   & 6.81   & 7.25   \\
& & Order  & 15 & 30325  & 34580  & 16807  & 17773  & 45.86  & 47.76  & 48.65  & 46.95  \\
\hline
\multirow{3}{*}{RndPauli-10}
& \multirow{3}{*}{502}
& Basic     & 19 & 487963 & 520824 & 251384 & 269440 & 100.00 & 100.00 & 100.00 & 100.00 \\
& & Flat   & 28 & 15264  & 24957  & 11325  & 12626  & 3.13   & 4.79   & 4.51   & 4.69   \\
& & Order  & 19 & 230358 & 261543 & 127532 & 134011 & 47.21  & 50.22  & 50.73  & 49.74  \\
\hline
\end{tabular}
}
\label{tab:LCU-benchmark1}
\end{table}

\textit{Discussion}. The results reveal a clear trade-off among circuit width, circuit depth, and overall gate count. 
In most cases, conditional flattening yields the largest reductions in depth and gate count, but increases circuit width in proportion to the control-register size. 
By contrast, the monotone-control decomposition approach (Order) introduces no additional control register and therefore preserves circuit width, making it more suitable for width-constrained architectures. 
Although its current performance is somewhat weaker, it could be further improved by combining it with conditional flattening to eliminate the remaining multi-controlled unitaries. 
Such a hybrid strategy may further reduce circuit depth, gate count, and ancilla resets while combining the strengths of both approaches.

\subsubsection{RQ2.2: Comparison with Existing Control-Flow Optimization}

To relate our structure-aware optimization to prior compiler work, we compare against Cline~\cite{wu2026cline}, which improves conditional flattening on Hamiltonian-style LCU benchmarks. The comparison is particularly relevant for evaluating the `Order' strategy. Cline adopts the general Hamiltonian $H_n = \sum_{P_j\in\{I,X,Y,Z\}^{\otimes n}} \alpha_j P_j$ as the LCU benchmark. In contrast, our approach explicitly identifies a basis set that covers all Pauli strings, enabling more efficient circuit construction and achieving better compression ratios, as shown in Table~\ref{tab:benchmark_cline_baseline}.
\begin{table}[htbp]
\centering
\caption{Benchmark comparison with Cline~\cite{wu2026cline} as the baseline.}
\resizebox{0.8\linewidth}{!}{
\setlength{\tabcolsep}{6pt}
\begin{tabular}{llrrrrrrrrrr}
\toprule
\multicolumn{2}{c}{} & \multicolumn{5}{c}{Count} & \multicolumn{5}{c}{Ratio vs. Cline (\%)} \\
\cmidrule(lr){3-7}\cmidrule(lr){8-12}
Name & Option & \#Depth & \#CX & \#H & \#S & \#T & Depth & \#CX & \#H & \#S & \#T \\
\midrule
\multirow{2}{*}{H3}
& Cline & 119 & 45  & 42  & 15 & 36 & 100.00 & 100.00 & 100.00 & 100.00 & 100.00 \\
& Order & 10  & 12  & 6   & 9  & 9  & 8.40   & 26.67  & 14.29  & 60.00  & 25.00  \\
\midrule
\multirow{2}{*}{H4}
& Cline & 158 & 60  & 56  & 20 & 48 & 100.00 & 100.00 & 100.00 & 100.00 & 100.00 \\
& Order & 10  & 16  & 8   & 12 & 12 & 6.33   & 26.67  & 14.29  & 60.00  & 25.00  \\
\midrule
\multirow{2}{*}{H6}
& Cline & 236 & 90  & 84  & 30 & 72 & 100.00 & 100.00 & 100.00 & 100.00 & 100.00 \\
& Order & 10  & 24  & 12  & 18 & 18 & 4.24   & 26.67  & 14.29  & 60.00  & 25.00  \\
\midrule
\multirow{2}{*}{H8}
& Cline & 314 & 120 & 112 & 40 & 96 & 100.00 & 100.00 & 100.00 & 100.00 & 100.00 \\
& Order & 10  & 32  & 16  & 24 & 24 & 3.18   & 26.67  & 14.29  & 60.00  & 25.00  \\
\bottomrule
\end{tabular}
}
\label{tab:benchmark_cline_baseline}
\end{table}

\paragraph{Remark: Generalizability to Hamiltonian Simulation.}
Although our framework is developed for open-system dynamics, its underlying
block-encoding formulation, together with LCU as a standard construction
primitive, also supports closed-system Hamiltonian simulation via QSVT~\cite{yuan2025cobblecompilingblockencodings}.
To demonstrate this broader applicability, we compare against representative
Trotter-based Hamiltonian simulation compilers~\cite{Campbell_2019,Li2022paulihedral,Cao2025marqsim}
on random Pauli Hamiltonians under a common error budget.

Figure~\ref{fig:qsvt_trotter} shows that our framework consistently requires
fewer gates at stringent error tolerances, in line with the more favorable
asymptotic scaling of QSVT ($\sim\log(1/\varepsilon)/\log\log(1/\varepsilon)$)
relative to Trotter-based methods ($\sim 1/\varepsilon$), albeit with a
higher upfront compilation cost.

\textit{Remark: Empirical Comparison with Cobble~\cite{yuan2025cobblecompilingblockencodings}.}
Cobble primarily studies the cost trade-off between LCU/Horner-style matrix-polynomial evaluation and QSVT, using unoptimized matrix-polynomial forms as the baseline; 
by contrast, our compiler already supports automatic expansion of matrix polynomials together with direct QSVT implementation, and therefore can faithfully reproduce the same functionality in this setting. 
Moreover, our compiler yields the same result on the QSVT benchmark in Cobble, as the instance implemented in Cobble uses Pauli $X$ as the underlying matrix, leaving essentially no additional space for structure-aware optimization. 

\begin{figure}
    \centering
    \includesvg[width=0.85\linewidth]{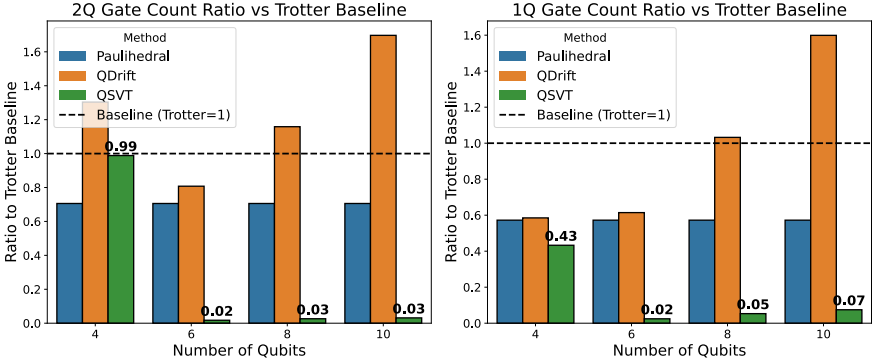}
    \caption{Gate count comparison between QSVT (ours) and Trotter-based
    Hamiltonian simulation compilers at a common error budget.
    Each bar shows the normalized fraction of gates relative to
    the unoptimized Trotter baseline.}
    \label{fig:qsvt_trotter}
\end{figure}

\subsection{RQ3: Correctness of the Lindbladian Frontend}
In this section, we investigate whether the Lindbladian frontend generates channel-simulation circuits whose empirical accuracy matches the theoretical error guarantees of the underlying short-time approximation.
\begin{figure}[b]
    \begin{minipage}[b]{0.45\textwidth}
        \raggedright
        \includesvg[width=0.95\textwidth]{Figures/cLCU_errs.svg} 
        \caption{Simulation error, measured by trace distance, of the channel-LCU algorithm compared with the theoretical upper bound. The evolution time is $\Delta t$.}
        \label{fig:err_channelLCU}
    \end{minipage}%
    \hspace{4mm}
    \begin{minipage}[b]{0.48\textwidth}
        \includesvg[width = 0.95\textwidth]{Figures/decay_only_Qmethod.svg}
        \caption{Simulation error of the higher-order expansion algorithm~\cite{Li2023simulating} for different expansion orders and decay parameters. The total evolution time is $T = 1$.}
         \label{fig:err_hoexp}
    \end{minipage}
\end{figure}

To evaluate correctness, we compare the error of the compiled simulation circuits against trusted reference evolutions and theoretical expectations. As shown in Fig.~\ref{fig:err_channelLCU}, the empirical error follows the predicted short-time behavior and remains within the theoretical upper bound across the tested time steps.

We further evaluate the higher-order expansion frontend in Fig.~\ref{fig:err_hoexp}. The results exhibit the expected trend that higher-order expansions achieve lower error, providing practical evidence that the frontend preserves the intended simulation accuracy.

\section{Related Works}
\paragraph{Intermediate representations in quantum compilers.}
In quantum compilers, most widely used toolchains remain centered on circuit-level abstractions~\cite{qiskit2024,CirqDevelopers_2025,Sivarajah_2021}, 
while several domain-specific systems raise the abstraction level for quantum simulation. For example, Paulihedral~\cite{Li2022paulihedral} uses structured Hamiltonian representations to optimize digital simulation beyond circuit level, and SimuQ~\cite{Peng2024simuq} formulates compilation for simulation kernels as a pattern-matching problem against analog quantum devices. 
Cobble~\cite{yuan2025cobblecompilingblockencodings} proposes the first language for programming with block-encodings and matrix polynomials, which supports direct optimization on the arithmetic structures of block-encoding operators. Our work continues this line of raising the compiler abstraction level, but does so around quantum channels: \ChannelIR{} treats Kraus- and Pauli-level structure as first-class compilation objects and serves as the basis of an end-to-end compiler for non-unitary dynamics.

\paragraph{Structured compilation and optimization.}
From a programming-languages perspective, compilation and optimization are often organized around structured intermediate representations and semantics-preserving rewrites, rather than purely local circuit transformations. Several works follow this line by manually identifying or automatically synthesizing~\cite{xu2022quartz, xu2023queso,xu2025optimize}
semantics-preserving rewrites for certain types of quantum programs, for example quantum control flows~\cite{Yuan2024tcomplexity,huang2024compiling} and uncomputation~\cite{sharma2025spare}, or even Dirac notation~\cite{xu2025automating}. 
Structured compilation has also been developed for Hamiltonian simulation and related block-encoding tasks. Several compilation frameworks for quantum simulation kernels~\cite{Li2022paulihedral,Peng2024simuq,Cao2025marqsim} propose high-level abstractions and intermediate representations (IR) for Hamiltonians. They leverage the expressive power of high-level abstractions to optimize digital quantum simulations beyond circuit level~\cite{Li2022paulihedral,Cao2025marqsim}, or formalize the compilation challenges of quantum simulations as pattern matching with analog quantum simulators~\cite{Peng2024simuq}, achieving the goal of optimizing resource overhead. As shown in our work, our \ChannelIR{} further elevates the abstraction level to quantum channels, achieves greater expressive power and serves as the foundation of an end-to-end compiler for the simulation of non-unitary dynamics. 

\paragraph{Quantum control-flow optimization methods.} 
Beyond generic quantum circuit optimizers~\cite{Amy_2019,Sivarajah_2021,xu2022quartz,xu2023queso,xu2025optimize}, several works optimize control-heavy quantum circuits more directly. One line of work simplifies control logic in a manner analogous to Boolean simplification~\cite{Schmitt2022tweedledum,Miller2003transform,qiskit2024,wu2026cline}, for example by rewriting truth-table representations or merging compatible control clauses. Another line treats control-flow optimization as the decomposition and rewriting of multi-controlled gates~\cite{li2022voqc,Yuan2024tcomplexity,sharma2025spare}, or reduces such overhead through specialized peripherals such as QROM~\cite{Babbush2018encoding}.
These techniques are most closely connected to our first optimization pass. Our second pass follows a structure-aware route at the IR level, using the Pauli-level algebraic information retained in \ChannelIR{} to optimize the implementation before the corresponding \SELECT{} multiplexors are lowered into circuit form.

\section{Data Availability Statement}

The repository containing the prototype implementation of the \ChannelIR{} and the evaluation scripts 
is currently available as an anonymous repository \cite{anon_lind}. 
This is a prototype code library for demonstration purposes only. 
The complete version of the code for reproduction will be provided during the artifact evaluation stage. 


\bibliographystyle{ACM-Reference-Format}
\bibliography{reference}

\clearpage
\appendix
\section{Proof Details for Propositions in \ChannelIR}

\begin{proof}[Proof of~\Cref{thm:soundness-of-rewrite-rules}]
It suffices to check each rewrite rule.

Rules \textbf{PS1} and \textbf{PS2} preserve the denotation of a Kraus expression by linearity:
\[
\llbracket \beta_1 P + \beta_2 P \rrbracket
= \beta_1 \llbracket P \rrbracket + \beta_2 \llbracket P \rrbracket
= (\beta_1+\beta_2)\llbracket P \rrbracket,
\]
and
\[
\llbracket 0 \cdot P \rrbracket = 0 = \llbracket \epsilon \rrbracket.
\]

For \textbf{K1}, the zero Kraus operator contributes nothing:
\[
0 \rho 0^\dagger = 0.
\]
Hence removing it does not change the channel.

For \textbf{K2}, a global phase cancels with its conjugate:
\[
(e^{i\theta}K)\rho(e^{i\theta}K)^\dagger
= e^{i\theta}K\rho e^{-i\theta}K^\dagger
= K\rho K^\dagger.
\]

For \textbf{C1}, the denotation of a channel is a sum over Kraus contributions, so permuting
the summands does not change the result.

For \textbf{C2}, let $K'_j=\sum_k u_{jk}K_k$ for a unitary matrix $U=(u_{jk})$.
Then for every density matrix $\rho$,
\[
\sum_j K'_j \rho K_j'^\dagger
=
\sum_j \left(\sum_k u_{jk}K_k\right)\rho
\left(\sum_\ell u_{j\ell}K_\ell\right)^\dagger
=
\sum_{j,k,\ell} u_{jk}\overline{u_{j\ell}}\, K_k \rho K_\ell^\dagger.
\]
Since $U$ is unitary, meaning that
\[
\sum_j u_{jk}\overline{u_{j\ell}} = \delta_{k\ell},
\]
the above simplifies to
\[
\sum_k K_k \rho K_k^\dagger.
\]
Therefore the transformed Kraus family denotes the same channel.

Rule \textbf{C2'} is the special case of \textbf{C2} for the $2\times 2$ unitary
\[
\begin{pmatrix}
a & b \\
-b^* & a^*
\end{pmatrix}.
\]

For \textbf{C3}, we compute
\[
\sum_{j=1}^{\ell} (a_jK)\rho(a_jK)^\dagger
=
\sum_{j=1}^{\ell}|a_j|^2\, K\rho K^\dagger
=
\left(\sum_{j=1}^{\ell}|a_j|^2\right)K\rho K^\dagger
=
\left(\sqrt{\sum_{j=1}^{\ell}|a_j|^2}\,K\right)\rho
\left(\sqrt{\sum_{j=1}^{\ell}|a_j|^2}\,K\right)^\dagger.
\]

Therefore every rewrite rule preserves denotation, and hence $C \Rightarrow_R C'$
implies $\llbracket C \rrbracket = \llbracket C' \rrbracket$.
\end{proof}

\begin{proof}[Proof of~\Cref{prop:kraus-rank-minimization}]
Let 
\[
C = \texttt{channel}\{K_1,\dots,K_n\}
\]
be a well-typed \ChannelIR{} program, and let $\mathcal{E} = \llbracket C \rrbracket$ denote its denotation. Let $r$ be the Kraus rank of $E$, i.e., the minimal number of nonzero Kraus operators required to represent $E$.

By definition of Kraus rank, there exists a Kraus representation
\[
\mathcal{E}(\rho) = \sum_{j=1}^r L_j \rho L_j^\dagger
\]
with $L_1,\dots,L_r \neq 0$, and no Kraus representation of $E$ uses fewer than $r$ nonzero operators. Extend this family to length $n$ by setting
\[
L_{r+1} = \cdots = L_n = 0.
\]
Then $\{L_1,\dots,L_n\}$ is also a valid Kraus representation of $E$.

Since both $\{K_1,\dots,K_n\}$ and $\{L_1,\dots,L_n\}$ are Kraus representations of the same channel, the unitary freedom of Kraus representations (\cite[Theorem~8.2]{NC10}) implies that there exists a unitary matrix $U = (u_{jk}) \in \mathbb{C}^{n \times n}$ such that
\[
L_j = \sum_{k=1}^n u_{jk} K_k, \qquad j = 1,\dots,n.
\]
Applying rewrite rule \textbf{C2}, which implements such a unitary mixing of Kraus operators, yields
\[
C_1 = \texttt{channel}\{L_1,\dots,L_r,0,\dots,0\}.
\]

Next, applying rule \textbf{K1} repeatedly removes the $n-r$ zero Kraus operators, producing
\[
C' = \texttt{channel}\{L_1,\dots,L_r\}.
\]
By Theorem~4.2, each rewrite rule preserves denotation, hence
\[
\llbracket C' \rrbracket = \llbracket C \rrbracket = \mathcal{E}.
\]
This establishes the existence of a rewrite sequence producing a Kraus representation with exactly $r$ nonzero operators.

For the rewrite complexity, note that rule \textbf{C3} is a special case of \textbf{C2}, so the above construction may be viewed as using one \textbf{C3} step, followed by one application of \textbf{C2}, and then at most $n-r$ applications of $K1$. This yields $O(n)$ total rewrites.

Alternatively, one may avoid a single $n \times n$ unitary by decomposing $U$ into a product of two-dimensional unitaries (by Givens rotations). Each such rotation acts on a pair of Kraus operators and is implemented by rule \textbf{C2'}, while diagonal phases are absorbed using rule \textbf{K2}. Using a standard elimination procedure, one can zero out the last $n-r$ Kraus operators using $O(n)$ applications of \textbf{C2'}, followed by at most $n-r$ applications of \textbf{K1}. 

This completes the proof.
\end{proof}
\section{Details in Structure-Aware Optimization}
\subsection{Proof details in Section VI}
We provide the proof of Proposition \ref{prop:existence} here. The correctness of \ref{prop:permute} is straightforward by 
noticing the definition of \SELECT and \PREPARE oracles, and the fact that permutation won't destroy the alignment of the coefficients and the unitaries.

\paragraph{Proof of Proposition \ref{prop:existence}}

\begin{proposition}[Existence of monotone-control decomposition]\label{prop:existence}
For an arbitrary set of $n$ unitary matrices $\{V_1, \cdots, V_n\}$, there exists a monotone-control decomposition $\{U_1, \cdots, U_n\}$ such that 
\(
\forall k \in \{1, \cdots, n\}, V_k = \Pi_{(l\leq_b k)} U_l.
\)

\end{proposition}
\begin{proof}
We prove this by induction on the Hamming weight of the binary representation of the number $k$, denoted as $w(k)$. 

$\bullet \ \textbf{Base case}$. For simplicity's sake, we set $V_0 = U_0 = I$. Then for any $ 0\leq d \leq \floor{\log n}$, $w(2^d) = 1$, so it is straightforward to determine that $U_{2^d} = P_{2^d}$.  

$\bullet \ \textbf{Inductive step}$. Now, assume that for all indices $l$ with $w(l) \leq m$, we can construct matrices $U_j$ for all $j \leq_b l$ such that $V_l = \Pi_{j\leq_b l} U_j$. For $w(k) = m + 1$, denote $S = \{l |l <_b k\}$, where $l <_b k$ means $l \leq_b k$ but $l \neq k$. Crucially for any $l \in S$, since $w(k) = m + 1$, $l$ can only have at most $m$ bits set to $1$, i.e., $w(l) \leq m$. Therefore $U_k$ can be directly calculated by an inverse of the decomposition: $U_k$ = $(\Pi_{l <_b k}U_l^{\dagger}) V_k$, since all of the $U_l$ with $w(l) \leq m$ has already been constructed. Notice that there might exist $n_s > n$ but $w(n_s) < w(n)$, (e.g. $n = 11$ but $n_s = 12$), but it won't bother since $U_{n_s}$ would have no influence on the construction of $U_n$, we can always view such $U_{n_s}$ and $V_{n_s}$ as `irrelevant terms'. 
Therefore the inductive step is completed. 
\end{proof}

\subsection{Algorithm details}

In the main text we only present a high-level overview of algorithm, designed to minimize the weighted-control cost of the Pauli-sum LCU.
Here we present the complete algorithm details, including the textual description and pseudocode for the main function and its subroutines. 

\paragraph{Pseudocode for algorithms}
\begin{algorithm}[t]
\caption{Find An Optimal Order of Pauli terms}
\label{alg:find-optimal-order}
\begin{algorithmic}[1]
\Require Pauli terms $\mathcal{T}=\{(P_i,\phi_i)\}_{i=1}^{m}$
\Ensure Mode table $\mathcal{M}_{\mathrm{add}}$, phase-corrected modes $\mathcal{M}_{G}$, control size $s$
\Statex \textbf{(Find Optimal Coverage Subspace)}
\State Separate identity term $I^{\otimes n}$ from $\mathcal{T}$ and record $\phi_0$
\State Convert each non-identity $P_i$ to symplectic vector \(v_i=(x_i\mid z_i)\in\mathbb{F}_2^{2n}\)
\State Build matrix $M$ (rows are $v_i$) and phase array $\Phi$; $s \gets \lceil \log_2 m \rceil$.

\State $(M_z,\Phi_z)\gets \textsc{Sort}(M,\Phi,\texttt{z})$; $(M_x,\Phi_x)\gets \textsc{Sort}(M,\Phi,\texttt{x})$

\Comment{primary: weight $\mathrm{w}_s$; secondary: tie-breaking on $x/z$ parts}

\State $(S_z,\mathcal{U}_z)\gets \textsc{GreedyBasisSelection}(M_x,s)$
\State $(S_x,\mathcal{U}_x)\gets \textsc{GreedyBasisSelection}(M_z,s)$

\If{$|\mathcal{U}_z|>|\mathcal{U}_x|$}
    \State $(M^\star,\Phi^\star,S^\star,\mathcal{U}^\star)\gets(M_x,\Phi_x,S_z,\mathcal{U}_z)$
\Else
    \State $(M^\star,\Phi^\star,S^\star,\mathcal{U}^\star)\gets(M_z,\Phi_z,S_x,\mathcal{U}_x)$
\EndIf

\Statex \textbf{(Assign Covered Paulis)}
\State $M_{\mathrm{sel}}\gets M^\star[S^\star]$; 

\State $\mathcal{M}_{\mathrm{sub}}\gets\{(\mathrm{Pauli}(v_i),\mathrm{bin}(2^{i-1}))\}_{i=1}^{d}$ \Comment{$v_i$ are rows of $M_{\mathrm{sel}}$, one-hot controls of length $s$}
\For{$u\in \mathcal{U}^\star\setminus\{b_1,\dots,b_d\}$}
  \State Decode $u=\bigoplus_i \alpha_i b_i$ by elimination, and set $c(u)=\bigoplus_i \alpha_i\,\mathrm{bin}(2^{i-1})$
  \State $\mathcal{M}_{\mathrm{sub}}\gets \mathcal{M}_{\mathrm{sub}}\cup\{(\mathrm{Pauli}(u),c(u))\}$
\EndFor
\Statex \textbf{(Assign Remaining Paulis)}
\State $\mathcal{R}\gets\{v\in M^\star:\ v\notin \mathcal{U}^\star\}$
\State $\mathcal{M}_{\mathrm{add}}\gets\textsc{AssignAdditionalModes}(s,\mathcal{M}_{\mathrm{sub}},\mathcal{R})$

\State Assign $I^{\otimes n}$ to control string $0^s$ 
\State Sort $\mathcal{M}_{\mathrm{add}}$ by integer value of control strings

\Statex \textbf{(Get Decomposed Paulis and Phases)}
\State $\mathcal{M}_{G}\gets\textsc{InvertModesWithPhases}(s,\mathcal{M}_{\mathrm{add}},\mathrm{\Phi}^\star,\phi_0)$

\State \Return $\mathcal{M}_{\mathrm{add}},\mathcal{M}_{G},s$
\end{algorithmic}
\end{algorithm}

\begin{algorithm}[t]
\caption{\textsc{AssignSubspaceModes} (compact)}
\begin{algorithmic}[1]
\Require selected generators $\{b_i\}_{i=1}^{d}$, covered set $\mathcal U_{\mathrm{cov}}$, control width $w$
\Ensure subspace mode set $\mathcal M$

\State $\mathcal M \gets \varnothing$
\For{$i=1,\dots,d$}
  \State assign one-hot control string $c_i \gets \mathrm{bin}(2^{i-1})$ (length $w$)
  \State $\mathcal M \gets \mathcal M \cup \{(\mathrm{Pauli}(b_i), c_i)\}$
\EndFor

\Statex Build pivot/elimination table from $\{(b_i,c_i)\}_{i=1}^d$ over $\mathbb F_2$

\For{each $u \in \mathcal U_{\mathrm{cov}}$ not already in $\{b_i\}$}
  \State $v \gets u,\; c \gets 0^w$
  \While{$v \neq 0$}
    \State find leading pivot of $v$, retrieve corresponding basis row $(r,\gamma)$
    \State $v \gets v \oplus r,\quad c \gets c \oplus \gamma$
  \EndWhile
  \State $\mathcal M \gets \mathcal M \cup \{(\mathrm{Pauli}(u), c)\}$
\EndFor

\State \Return $\mathcal M$
\end{algorithmic}
\end{algorithm}

For the one-hot generators extracted from $\mathcal M_{\mathrm{sub}}$, we write
$b(u)=\bigoplus_{i:\,u_i=1} b_i$ for each $u\in\{0,1\}^d$.

\begin{algorithm}[t]
\caption{\textsc{AssignAdditionalModes} (compact)}
\begin{algorithmic}[1]
\Require subspace mode set $\mathcal M_{\mathrm{sub}}$, remaining vectors $\mathcal R$, control width $s$
\Ensure completed mode table $\mathcal M_{\mathrm{add}}$

\State $\mathcal M_{\mathrm{add}} \gets \mathcal M_{\mathrm{sub}}$
\State Extract the one-hot entries in $\mathcal M_{\mathrm{sub}}$ as
$\{(\mathrm{Pauli}(b_i),\mathrm{bin}(2^{i-1}))\}_{i=1}^{d}$
\State $r \gets s-d$

\For{$m=1,\dots,2^r-1$}
  \If{$\mathcal R=\varnothing$}
    \State \textbf{break}
  \EndIf
  \State $p \gets \mathrm{bin}(m)$ on $r$ bits, and $\mathcal C \gets \{0,1\}^{d}\setminus\{0^d\}$

  \If{$|\mathcal R|>(2^r-m+1)(2^d-1)$}
    \State choose $v_0\in\arg\min_{v\in\mathcal R}\mathrm{w}_s(v)$ and assign it to $p\|0^d$
    \State use $v_0\oplus b(u)$ as the reference for each $u\in\mathcal C$
    \State remove $v_0$ from $\mathcal R$
  \Else
    \State use $b(u)$ as the reference for each $u\in\mathcal C$
  \EndIf

  \Statex Greedily match unused low-bit strings to remaining vectors
  \While{$\mathcal R\neq\varnothing$ and some $u\in\mathcal C$ is unused}
    \For{each $v\in\mathcal R$}
      \State choose an unused $u(v)\in\arg\min_{u\in\mathcal C}\mathrm{w}_s(v\oplus \mathrm{ref}(u))$
    \EndFor
    \State pick the pair $(v^\star,u^\star)$ with globally minimum mismatch
    \State $\mathcal M_{\mathrm{add}} \gets \mathcal M_{\mathrm{add}}\cup\{(\mathrm{Pauli}(v^\star),p\|u^\star)\}$
    \State remove $v^\star$ from $\mathcal R$ and mark $u^\star$ used
  \EndWhile
\EndFor

\State Let $\mathcal F$ be the set of unused nonzero control strings of length $s$
\While{$\mathcal R\neq\varnothing$ and $\mathcal F\neq\varnothing$}
  \State sort $\mathcal F$ by
  $\bigl(\mathrm{hw}(c),\sum_{l<\mathrm{int}(c)}\mathrm{w}_s(P_l),\mathrm{int}(c)\bigr)$
  \State choose the first $c\in\mathcal F$, and let $u$ be its last $d$ bits
  \State choose $v^\star\in\arg\min_{v\in\mathcal R}\mathrm{w}_s(v\oplus b(u))$
  \State $\mathcal M_{\mathrm{add}} \gets \mathcal M_{\mathrm{add}}\cup\{(\mathrm{Pauli}(v^\star),c)\}$
  \State remove $v^\star$ from $\mathcal R$ and remove $c$ from $\mathcal F$
\EndWhile

\State \Return $\mathcal M_{\mathrm{add}}$
\end{algorithmic}
\end{algorithm}

\begin{algorithm}[t]
\caption{\textsc{InvertModesWithPhases} (compact)}
\begin{algorithmic}[1]
\Require mode table $\mathcal M_{\mathrm{add}}$, phase lookup $\Phi^\star$, zero-term phase $\phi_0$, control width $s$
\Ensure phase-corrected productive modes $\mathcal M_G$

\State For each control string $b\in\{0,1\}^s$, initialize target Pauli $P_b\gets I^{\otimes n}$,
target phase $\varphi_b\gets 0$, productive Pauli $g_b\gets I^{\otimes n}$,
productive phase $\theta_b\gets 0$, and relevance flag $r_b\gets 0$

\ForAll{$(P,b)\in\mathcal M_{\mathrm{add}}$}
  \If{$P$ appears in the original Pauli set}
    \State $P_b\gets P$, $\varphi_b\gets \Phi^\star(P)$, $r_b\gets 1$
  \ElsIf{$P=I^{\otimes n}$ and $\phi_0\neq 0$}
    \State $P_b\gets I^{\otimes n}$, $\varphi_b\gets \phi_0$, $r_b\gets 1$
  \EndIf
\EndFor
\State $\varphi_{0^s}\gets \phi_0$

\For{each relevant address $b$ in increasing integer order}
  \If{$\mathrm{hw}(b)=1$}
    \State $g_b\gets P_b$, $\theta_b\gets (\varphi_b-\phi_0)\bmod 4$
  \Else
    \State $Q\gets P_b$, $\theta\gets \varphi_b$
    \For{each $c<_b b$ in increasing integer order}
      \If{$g_c\neq I^{\otimes n}$ or $\theta_c\not\equiv 0\pmod 4$}
        \State remove the contribution of $g_c$ from $Q$, and add the induced overlap phase to $\theta$
      \EndIf
    \EndFor
    \State $g_b\gets Q$, $\theta_b\gets \theta\bmod 4$
  \EndIf
\EndFor

\Statex Reconstruct the phase implied by the productive modes
\For{each $b\in\{0,1\}^s$}
  \State multiply all nontrivial $g_c$ with $c\leq_b b$ in increasing integer order
  \State denote the resulting accumulated phase by $\widehat{\varphi}_b$
\EndFor

\State $\phi_{\mathrm{ad}}(b)\gets(\varphi_b-\widehat{\varphi}_b)\bmod 4$ for all $b$
\ForAll{$(P,b)\in\mathcal M_{\mathrm{add}}$}
  \State $\theta_b\gets \theta_b+\phi_{\mathrm{ad}}(b)$
\EndFor

\State $\mathcal M_G\gets \{(g_b,b,\theta_b): g_b\neq I^{\otimes n}\text{ or }\theta_b\not\equiv 0\pmod 4\}$
\State \Return $\mathcal M_G$
\end{algorithmic}
\end{algorithm}

\paragraph{Text Description}
We present text description for these subroutines that appear in the main function:

 $\bullet\ $ \textsc{Sort}: We sort the matrix whose rows are Pauli symplectic vectors. The primary key is the Pauli weight $\mathrm{w}_s$; ties are broken by the $x/z$ parts according to the chosen criterion.

$\bullet\ $ \textsc{GreedyBasisSelection}: We incrementally construct a generator set $\mathcal{M}_{sub}$ in symplectic space. At each step, each candidate is scored by subspace-growth gain per Pauli weight; we select the best one and update the generated subspace. The procedure stops when no positive gain remains, when $|\mathcal{M}_{sub}|=s$, or when the remaining control-address capacity cannot cover all uncovered Paulis.

$\bullet\ $ \textsc{AssignAdditionalModes}: After the subspace is fixed, uncovered Paulis are assigned to remaining control addresses by a hierarchical greedy rule: we first use unused high-bit prefixes and, within each prefix, choose low-bit values that minimize weight mismatch. If nonzero low-bit values are insufficient, we additionally allow the zero low-bit value with a lightweight reference offset; any leftover terms are then filled into remaining free addresses by the same minimum-mismatch principle.

$\bullet\ $ \textsc{InvertModesWithPhases}: After all Pauli positions are fixed, we invert
\(
P_j=\prod_{l\leq_b j} g_l\)
to recover the $g_l$ operators and their phases. The procedure builds a target $P$-table and then reconstructs the $g$-table by removing previously accumulated contributions while accounting for overlap phases from non-commuting Pauli products, ensuring consistency with the original Pauli set.

\section{Details About Compilation Accuracy }


To address the last research question in the main text, we designed a lightweight experiment to verify the accuracy of the compiled circuits generated by the Lindbladian frontend. 
Here in the appendix we provide comprehensive details of the experiment. 

\subsection{Round I: Channel-LCU Circuits} 
According to~\cite{cleve2017lindblad}, the error produced by the channel-LCU circuit $\mathcal{M}_\delta$ simulating the Lindblad evolution $e^{\delta \mathcal{L}}$ should obey the upper bound as follows: 
\[
\|\mathcal{M}_\delta - e^{\mathcal{L}\delta}\|_\diamond \leq 5(\delta\|\mathcal{L}\|_{ops})^2,
\]
where 
\[
\mathcal{L}_{ops} = \|H\| + \sum_{j=1}^{m}\|L_j\|^2 \ \textrm{for} \ \mathcal{L}\rbra{\rho} = -\ii [\rho, H] + \sum_j \rbra*{L_j\rho L_j^{\dagger} - \frac{1}{2}\cbra*{L_j^{\dagger}L_j, \rho}}. 
\]
We use the Transversal-field Ising model (TFIM) in Example \ref{example:tfim} as the test benchmark, and use both the superoperator from direct calculation (for $N \leq 6$) and simulation results from QuTiP~\cite{qutip5}'s \textit{mesolve} (a classical differential equation solver) (for $N = 7$) as baseline results. We compare the simulated error with the theoretical error bound, which is shown in Figure \ref{fig:err_channelLCU}. 
We also illustrate the simulation time of both QuTiP mesolve and our method in Table \ref{tab:sim_time_clcu}; 
we observe that a TFIM model with \(N = 8\) caused QuTiP's `mesolve' function to crash due to exceeding memory limits. In contrast, our method accommodates \(N = 12\) while also requiring less simulation time.
\begin{figure}[b]
    \begin{minipage}[b]{0.6\textwidth}
        \raggedright
        \includesvg[width=0.9\textwidth]{Figures/cLCU_errs.svg} 
        \captionof{figure}{Simulation error (trace distance) compared with theoretical upper bounds; The evolution time is $\Delta t$.}
        \label{fig:err_channelLCU_ap}
    \end{minipage}%
    \hspace{4mm}
    \begin{minipage}[b]{0.3\textwidth}
    \flushleft
        \begin{tabular}{lcc}
            \toprule
            $N$ & mesolve~\cite{qutip5} & Ours \\
            \midrule
            3 & $0.0013s$ & $0.09s$ \\
            5 & $0.30s$  & $0.22s$ \\
            7 & $904.49s$ & $0.89s$ \\
            8 & $N/A$  &  $3.03s$\\
            10 &  $N/A$ & $17.77s$ \\
            12 & $N/A$ & $188.03s$ \\
            \bottomrule
        \end{tabular}
        \vspace{8ex}
        \vfill
        \captionof{table}{Simulation time cost for mesolve and our method.}
         \label{tab:sim_time_clcu_ap}
    \end{minipage}
\end{figure}

\subsection{Round II: Circuits from Higher-order Expansions}
The construction of the quantum channel based on the higher-order expansion of the Lindbladian evolution operator is way more complex than the channel-LCU for the first-order approximation.
We provide an insight derived from Duhamel's principle and directly provide the explicit construction of the channel for the $k$-th-order expansion. For more details please refer to \cite{Li2023simulating}. 
We test the accuracy of the circuits generated from the higher-order expansion Lindbladian frontend. The quantum channel is generated by expanding the evolution operator $e^{\mathcal{L}t}$ to higher orders following Duhamel's principle. 
\begin{definition}[Quantum channel for $k$-th order expansion of Lindbladian evolution operator]
For a differential equation with the form of $u'(t) = Lu + f(t, u(t))$, Duhamel's principle states that the solution can be expressed as:
\[
u(t) = e^{tL}u_0 + \int_0^t e^{(t-s)L}f(s,u(s))ds
\]
Let $\mathcal{L}$ be a Lindbladian superoperator, then $\mathcal{L}$ can be similarly decomposed into a drifting part $\mathcal{L}_D$ and a jump part $\mathcal{L}_J$, 
where 
\[
\mathcal{L}_D(\rho) = -\ii [\rho, H] - \frac{1}{2}\sum_j \cbra*{L_j^{\dagger}L_j, \rho} = J\rho + \rho J^{\dagger}, \quad \mathcal{L}_J(\rho) = \sum_j L_j\rho L_j^{\dagger}.
\]
Applying Duhamel's principle first time we obtain: 
\[
\rho_t = e^{\mathcal{L}_Dt}(\rho_0) + \int_0^t e^{\mathcal{L}_D(t-s)}(\mathcal{L}_J\rho_s)ds
\]
By iteratively applying Duhamel's principle, we can derive the $k$-th-order expansion of the solution as follows:

\[
\begin{aligned}
\rho_t &= e^{\mathcal{L}_Dt}(\rho_0)  \\
&+ \sum_{k=1}^{K}\int_{0\leq s_1 \leq \cdots \leq s_k \leq t} e^{\mathcal{L}_D(t-s_k)}\mathcal{L}_J e^{\mathcal{L}_D(s_k-s_{k-1})}\cdots \mathcal{L}_J e^{\mathcal{L}_Ds_1}(\rho_0) ds_1\cdots ds_k \\
& + O(t^{K+1}) 
\end{aligned}
\]

Therefore we construct the Kraus operators as follows: 
$A_0 = e^{Jt}$, and for $j \geq 1$,

\[
A_j = \sqrt{\hat{w}_{j_k}\cdots\hat{w}_{(j_k, \cdots, j_1)}} e^{J (t - \hat{x}_{j_k})}L_{l_k} \cdots e^{J(\hat{x}_{(j_k,\cdots,j_2)} - \hat{x}_{(j_k, \cdots, j_1)})}L_{l_1}e^{J\hat{x}_{(j_k, \cdots,j_1)}},
\]

where $\hat{w}_{j_k}\cdots\hat{w}_{(j_k, \cdots, j_1)}$ is the weight of the $k$-th order expansion term, and $\hat{x}_{(j_k, \cdots, j_1)}$ is the sampled quadrature time points for the $k$-th-order expansion term. 
Both sets of values are generated by the quadrature method, and we use the Gauss-Legendre quadrature method in our implementation.

We further use a $K'$-th-order expansion of the drifting part $e^{Jt} = \sum_{k=0}^{K'}\frac{(Jt)^k}{k!}$, 
and the final quantum channel is constructed as $\mathcal{M}(\rho) = \sum_j A_j \rho A_j^{\dagger}$.
\end{definition}

Then we test the simulation accuracy of the circuit generated from the above quantum channel construction. We apply the decay model in Section 3 as the test benchmark, 
and use the superoperator from direct calculation as the baseline result. We divide the total evolution time $T = 1$ into $M$ segments,
simulate the evolution for one segment, and then calculate the trace distance $dist = \|\rho_{\text{sim}} - \rho_{\text{baseline}}\|_1$ between the simulated state and the baseline state.

The final error rate is calculated as $dist * M$, which is the upper bound of the error rate for the whole evolution.
The results are shown in Figure \ref{fig:decay_qmethod}.
\begin{figure}
    \centering
    \includegraphics[width=0.9\linewidth]{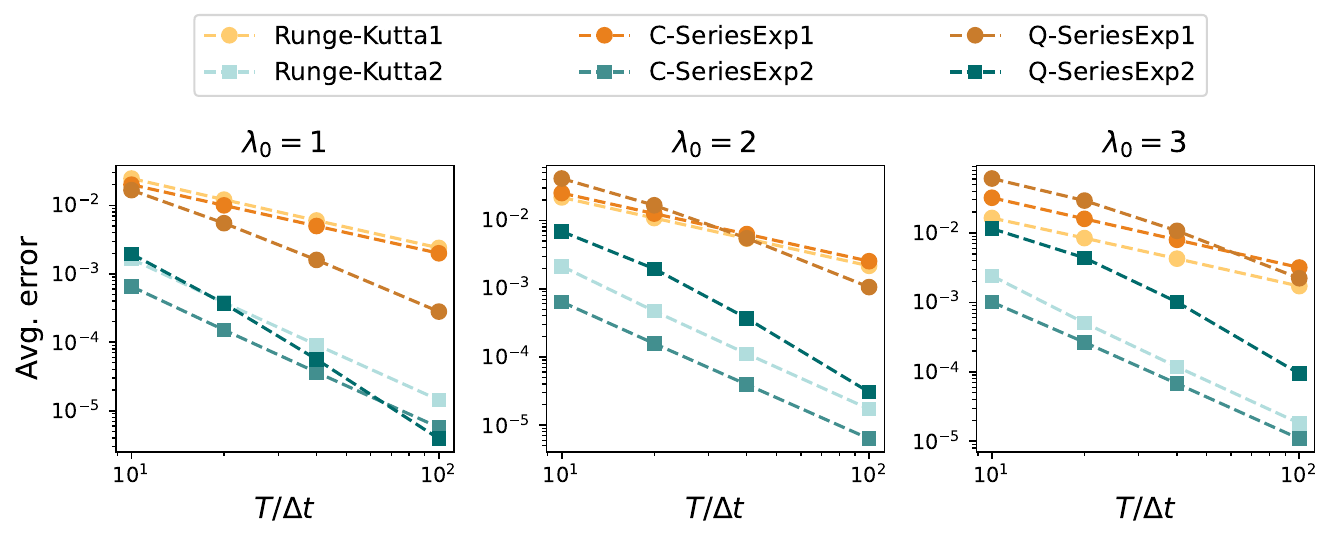}
    \caption{Simulation accuracy of the quantum channel generated from higher-order expansion Lindbladian frontend compared with the classical baseline.}
    \label{fig:decay_qmethod}
\end{figure}
\end{document}